\documentclass[letterpaper, 10 pt, conference]{ieeeconf}
\usepackage{amsmath}
\usepackage{amsfonts,amssymb}
\usepackage{tikz} 
\usepackage{xr}
\usepackage{algorithm}
\usepackage{algorithmicx}
\usepackage{algpseudocode}
\usepackage{cite}
\usepackage{bbm}
\usepackage[caption=false]{subfig}

\newtheorem{theorem}{Theorem}
\newtheorem{lemma}[theorem]{Lemma}

\newtheorem{remark}[theorem]{Remark}

\IEEEoverridecommandlockouts   
\IEEEoverridecommandlockouts                              
\title{\LARGE \bf
A Bernoulli-Gaussian Physical Watermark for Detecting Integrity Attacks in  Control Systems 
}
\author{Sean Weerakkody~~~~~~~~Omur Ozel~~~~~~~~Bruno Sinopoli
\thanks{S Weerakkody, O. Ozel, and B. Sinopoli are with the Department of Electrical and Computer Engineering, Carnegie Mellon University, Pittsburgh, PA, USA 15213. Email: {\tt\small \{sweerakk, oozel\}@andrew.cmu.edu, brunos@ece.cmu.edu}}
\thanks{S. Weerakkody is supported in part by the Department of Defense (DoD) through the National Defense Science \& Engineering Graduate Fellowship (NDSEG) Program. The work by S. Weerakkody,  O. Ozel, and B. Sinopoli is supported in part by the Department of Energy under Award Number DE-OE0000779 and by the National Science Foundation 
under award number CCF 1646526.}}

\begin{document} \maketitle
\begin{abstract}
We examine the merit of Bernoulli packet drops in actively detecting integrity attacks on control systems. The aim is to detect an adversary who delivers fake sensor measurements to a system operator in order to conceal their effect on the plant. Physical watermarks, or noisy additive Gaussian inputs, have been previously used to detect several classes of integrity attacks in control systems. In this paper, we consider the analysis and design of Gaussian physical watermarks in the presence of packet drops at the control input. On one hand, this enables analysis in a more general network setting. On the other hand, we observe that in certain cases, Bernoulli packet drops can improve detection performance relative to a purely Gaussian watermark. This motivates the joint design of a Bernoulli-Gaussian watermark which incorporates both an additive Gaussian input and a Bernoulli drop process. We characterize the effect of such a watermark on system performance as well as attack detectability in two separate design scenarios. Here, we consider a correlation detector for attack recognition. We then propose efficiently solvable optimization problems to intelligently select parameters of the Gaussian input and the Bernoulli drop process while addressing security and performance trade-offs. Finally, we provide numerical results which illustrate that a watermark with packet drops can indeed outperform a Gaussian watermark.
\end{abstract}
\section{Introduction}
The security of cyber-physical systems (CPS) has become a critical issue \cite{Cardenas:2008ke}. Since CPS such as the smart grid, waste management systems, water distribution systems, transportation systems, and smart buildings are linked to critical infrastructures, it is imperative that they operate securely. Unfortunately, attacks have occurred against CPS. This includes Stuxnet \cite{Chen2010}, which targeted uranium enrichment facilities in Iran, the Maroochy Shire incident \cite{Slay2008}, an attack by a malicious insider on a sewage management system, and the Ukraine power attack \cite{liang20172015}, a hack resulting in widespread blackouts in Ukraine. The threat does not appear to be over as the growing connectivity and heterogeneity of our system architectures provide new attack surfaces for adversaries.

We focus on detecting integrity attacks in control systems in the presence of packet drops at the control input. In an integrity attack, an adversary modifies inputs and sensor measurements in a control system. The goal of such an attacker may be to achieve some economic benefit or cause physical damage to a system. An attacker can potentially maximize his impact by hiding his presence from the operator. Remaining stealthy allows an attacker to affect the system for long periods of time without defender interference. The adversary can avoid detection by intelligently modifying the sensor measurements to fool detectors. One example is a replay attack, as used in Stuxnet, where an attacker replaces true outputs with a previously recorded sequence of measurements \cite{Mo2009R}. 

We consider the use of physical watermarking to detect integrity attacks. A physical watermark is a noisy Gaussian signal added on top of an optimal control input to authenticate a system's dynamics. Physical watermarking is a method of active detection, where a defender alters his strategy to recognize attacks. These methods are necessary when standard fault detection methods provably fail \cite{bai2015security}, \cite{weerakkody2016informationflow}. Recent work has investigated physical watermarking. In \cite{Mo2009R,Chabukswar2013,Mo2014,miao2013stochastic}, the design of watermarks against replay attacks was examined. Additionally, \cite{satchidanandan2017dynamic} and \cite{hespanhol2017dynamic} design asymptotic detectors in systems implementing physical watermarking to ensure zero additive distortion power is introduced into sensor measurements. Additionally, in a scalar setting, \cite{hosseini2016designing} demonstrates the optimality of Gaussian watermarks against Gaussian attackers and vice versa. \cite{rubio2017use} evaluates the use of non-stationary watermarks to hamper system identification. Finally, \cite{Weerakkody2014} considers watermarks to thwart adversaries who have access to a subset of inputs and model knowledge.

However, prior work fails to consider the scenario where there exists packet drops in the network. In this paper, we generalize the design of the Gaussian physical watermark by incorporating Bernoulli drops at the control inputs. This enables the operator to account for imperfect networks when designing a Gaussian watermark for secure detection. We also argue that using Bernoulli drops together with a Gaussian watermark can improve detection. This motivates the analysis and design of a joint Bernoulli-Gaussian watermark. In our preliminary work \cite{omur_smartgridcomm}, we proposed using packet drops in a setting without Gaussian watermarks to detect replay attacks. This article extends these results by providing a rigorous mathematical setting to jointly design parameters of both a Gaussian watermark and Bernoulli drop process. 

We investigate two types of watermark design: 1) a watermark  with an  independent and identically distributed (IID) Gaussian additive input multiplied by a Markovian Bernoulli drop process at the control input and 2) a watermark with a stationary Gaussian additive input generated by a hidden Markov model (HMM) multiplied by an IID Bernoulli drop process at the control input. We incorporate a correlation detector \cite{Chabukswar2011},\cite{Chabukswar2013} to recognize integrity attacks and characterize adversarial scenarios where the Bernoulli-Gaussian watermark is provably effective. Next, we provide efficiently solvable optimization problems to design parameters of the Gaussian input and the Bernoulli drop process. Simulation results illustrate scenarios where packet drops improve detection performance relative to a purely Gaussian watermark.

\section{System Model} \label{Model}
We consider a discrete time LTI control system as follows
\begin{equation}
x_{k+1} = Ax_k + Bu_{k,c} + w_k, ~~ y_k = Cx_k + v_k. \label{eq:openloop}
\end{equation}
 $x_k \in \mathbb{R}^n$ is the state vector at time $k$.  A set of $m$ sensor measurements $y_k \in \mathbb{R}^m$ is delivered to a supervisory control and data acquisition (SCADA) system at time $k$ in order to perform remote estimation and compute an intended control input $u_k \in \mathbb{R}^p$. A set of $p$ control inputs $u_{k,c} \in \mathbb{R}^p$ actuate the system. We differentiate $u_{k,c}$, the control input applied to the system, versus $u_k$, the input computed by a SCADA operator. We assume $w_k \sim \mathcal{N}(0,Q)$ is IID process noise and $v_k \sim \mathcal{N}(0,R)$ is IID measurement noise (independent of $\{w_k\}$), where $Q \succ 0, R \succ 0$. A Kalman filter performs state estimation as follows.
\begin{align}
&\hat{x}_{k+1|k} = A\hat{x}_{k|k} + Bu_{k,c}, ~\hat{x}_{k|k} = \hat{x}_{k|k-1} + Kz_k, \label{eq:KalmanEst}   \\
&K = PC^T(CPC^T+R)^{-1},~ z_k = y_k - C\hat{x}_{k|k-1}, \\
 &P= APA^T + Q - APC^T(CPC^T+R)^{-1}CPA^T.
 \end{align}
The defender minimizes a cost function $J$:
\begin{equation}
J = \lim_{N \rightarrow \infty} \frac{1}{2N+1} \mathbb{E} \left[ \sum_{k = -N}^N x_k^TWx_k + u_{k,c}^T U u_{k,c} \right],
\end{equation}
where $W \succ 0$ and $U \succ 0$. We assume $(A,B)$ and $(A,Q^{\frac{1}{2}})$ are controllable and $(A,C)$ and $(A,W^{\frac{1}{2}})$ are observable.

\subsection{Control in Uncertain Networks}
As shown in Fig. \ref{fig:model}, the control input $u_k$ may be dropped as it is sent from the SCADA system to the plant. Here, 
\begin{equation}
u_{k,c} = \eta_k u_k,
\end{equation} 
where $\eta_k \in \{0,1\}$ is a Bernoulli random variable. The control input $u_k$ may be dropped due to network imperfections. In this case, we assume the operator receives an acknowledgement (ACK), which specifies if $u_k$ was delivered. Alternatively, the input $u_k$ may be intentionally dropped as a means to watermark the system, enabling the detection of integrity attacks that fail to preserve the effect of the drop process. This strategy was initially investigated in \cite{omur_smartgridcomm}. We consider both IID and Markovian drop processes.

\begin{figure}[h]
\centerline{\includegraphics[height=.65\linewidth]{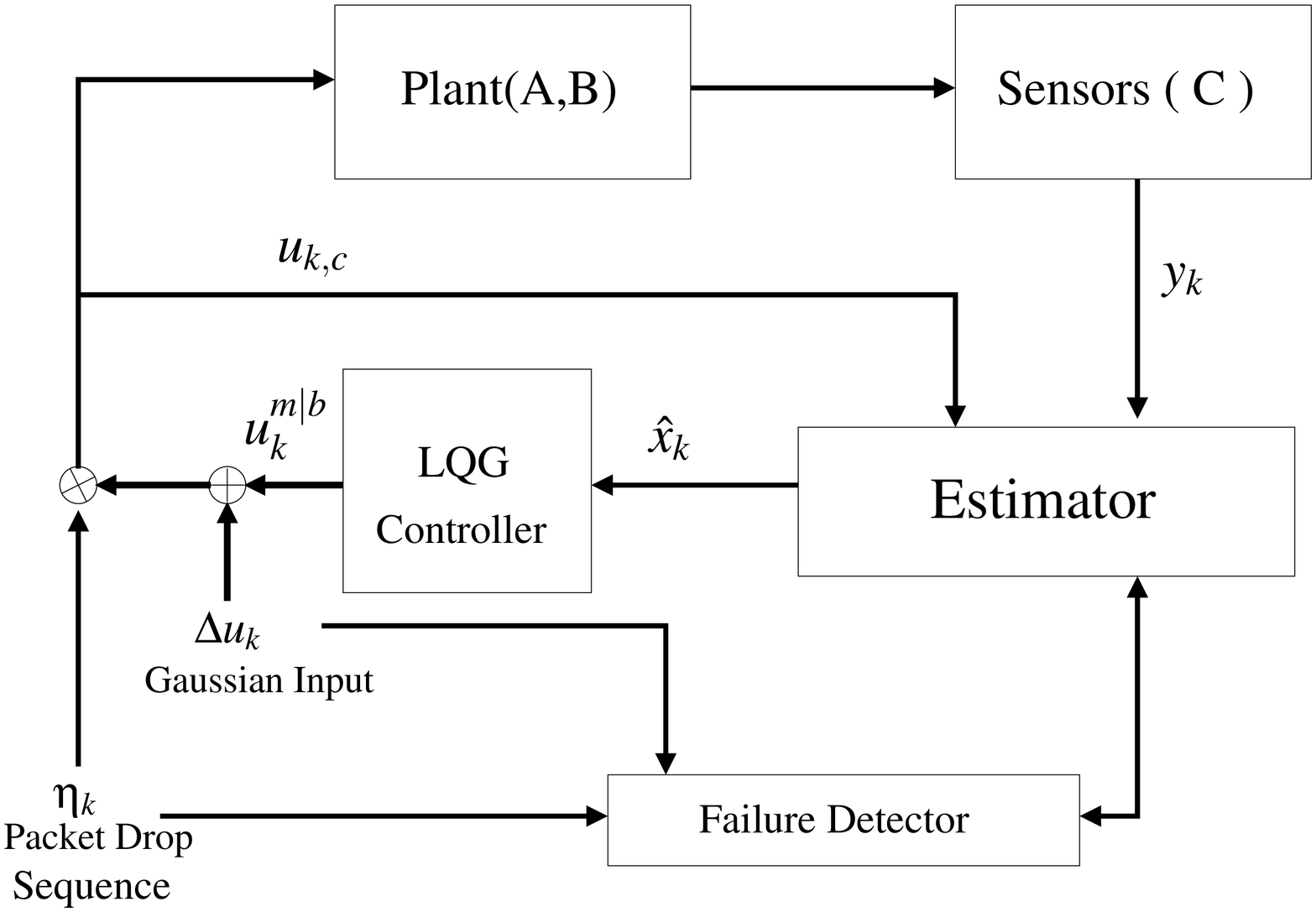}}
\caption{System model}
\label{fig:model}
\end{figure}

\subsubsection{IID Bernoulli Process}
First, we assume $\{\eta_k\}$ is an IID Bernoulli process where $P(\eta_k = 1) = 1-p_d$. LQG control with IID Bernoulli packet losses was studied in \cite{schenato2007foundations}.  Consider the information set $\mathcal{F}_k \triangleq \{y_{-\infty:k}, \eta_{-\infty:k-1}, u_{-\infty:k-1}\}$. We suppose $p_d$ is chosen (or given) so that the  system \eqref{eq:openloop} can have finite cost $J$. The optimal control strategy at time $k$ given $\mathcal{F}_k$ is as follows \cite{schenato2007foundations,omur_smartgridcomm}:
\begin{align}
u_k^b &= L_k \hat{x}_{k|k},~L_k = -(B^TS_{k+1}B+U)^{-1}B^TS_{k+1}A,  \nonumber \\
S_k &= A^TS_{k+1}A  + W + (1-p_d)A^TS_{k+1}BL_k. 
\end{align}
As we expect the system has been running for a long time, both $L_k$ and $S_k$ have converged to fixed point values so that
\begin{align}
u_k^b &= L_{(b)} \hat{x}_{k|k},~L_{(b)} = -(B^TS_{(b)}B+U)^{-1}B^TS_{(b)}A, \nonumber \\
S_{(b)} &= A^TS_{(b)}A  + W  + (1-p_d)A^TS_{(b)}B L_{(b)}. 
\end{align}
$J =J_{(b)}$ for this strategy where $J_{(b)}$ is
\begin{equation}
J_{(b)} = \mbox{tr}\left(S_{(b)} Q +(A^TS_{(b)}A+W-S_{(b)})(P-KCP)\right).
\end{equation}

\subsubsection{Markovian Bernoulli Process}
In this setup, we assume there are Markovian packet losses \cite{mo2013lqg} at the input where
\begin{equation}
\begin{bmatrix} P(\eta_{k+1}=0|\eta_k=0) P(\eta_{k+1}=1|\eta_k=0) \\ P(\eta_{k+1}=0|\eta_k=1) P(\eta_{k+1}=1|\eta_k=1) \end{bmatrix} = \begin{bmatrix} \bar{\alpha} & \alpha \\ \beta & \bar{\beta} \end{bmatrix} \label{eq:markov}
\end{equation}
and $\bar{\alpha} \triangleq 1-\alpha$, $\bar{\beta} \triangleq 1-\beta$. Here, we assume $0 < \alpha \le 1$, $0 < \beta \le 1$ so that $\eta_k$ is irreducible. Moreover, we assume $\eta_k$ is stationary, which can be obtained by letting its initial distribution be $P(\eta_{-\infty} = 0) = \frac{\beta}{\alpha+\beta}$. Finally, we assume that $\alpha$ and $\beta$ are selected (or given) so that the system \eqref{eq:openloop} can have finite cost $J$. The optimal strategy at time $k$ given $\mathcal{F}_k$ is 
\begin{align}
u_k^m &= L_{(m)} \hat{x}_{k|k},~L_{(m)}  = -(B^T R_{(m)} B + U)^{-1} B^T R_{(m)} A, \nonumber \\
R_m &= A^T(\beta S_{(m)} + \bar{\beta}R_{(m)})A + W + \bar{\beta}A^T R_{(m)} B L_{(m)} \nonumber, \\
S_m &= A^T(\bar{\alpha} S_{(m)} + \alpha R_{(m)})A + W  + \alpha A^T R_{(m)} B L_{(m)} \nonumber,
\end{align}
where $L_{(m)}, R_{(m)}, S_{(m)}$ are parameters which converged to their steady state values. The resulting cost of control is
\begin{align}
&J = J_{(m)} =  \frac{\mbox{tr}(\beta S_{(m)} Q + \alpha R_{(m)} Q)}{\alpha+\beta} \\
& + \frac{\mbox{tr}((A^T(\bar{\alpha}S_{(m)} + \alpha R_{(m)})A + W - S_{(m)})(P-KCP))}{\alpha+\beta} \nonumber.
\end{align}
\begin{remark}
The prior strategies are optimal when the defender only has knowledge of the observed drop sequence $\eta_{-\infty:k-1}$. However, if the drop sequence is intentionally introduced using a pseudo random number generator (PRNG), the defender knows future values of $\eta_k$. The design of a controller that uses this information is left for future work.
\end{remark}
\subsection{Joint Bernoulli-Gaussian Physical Watermarking}
To account for adversarial behavior, we consider additive Gaussian physical watermarks $\Delta u_k$. First introduced in \cite{Mo2009R} to detect replay attacks, an additive Gaussian watermark can be leveraged to verify the freshness of outputs. We aim to intelligently combine the Gaussian watermarks considered in \cite{Chabukswar2013} and \cite{Mo2014} with a Bernoulli drop process at the input. Such a design accomplishes two goals: 1) to expand the analysis of physical watermarking to a more realistic network setting with packet drops and 2) to potentially improve performance by considering a more general joint Bernoulli-Gaussian watermark.

We consider two main joint designs. \\
\textbf{Watermark 1: IID Gaussian Input + Markovian Drops}
\begin{equation}
u_{k,c} = \eta_k(u_k^m + \Delta u_k). \label{eq:markovian input}
\end{equation}
 $\{\eta_k\}$ is a Markovian Bernoulli process and $\Delta u_k \sim \mathcal{N}(0,\mathcal{Q})$ is an IID Gaussian watermark \cite{Mo2009R}. We assume $\Delta u_k$ is independent of other stochastic processes in the system.\\
\textbf{Watermark 2: Stationary Gaussian Input + IID Drops} \\
\begin{equation}
u_{k,c} = \eta_k(u_k^b + \Delta u_k).
\end{equation}
In this case, $\{\eta_k\}$ is an IID Bernoulli process. The Gaussian input $\Delta u_k$ is assumed to be a stationary process generated by a hidden Markov model (HMM) as considered in \cite{Mo2014}.
\begin{equation}
\zeta_{k+1} = A_\omega \zeta_k + \psi_k, ~~~ \Delta u_k = C_h \zeta_k \label{eq:stationary}.
\end{equation}
$\zeta_k$ is the hidden state of the HMM, $A_{\omega}$ has spectral radius $\rho(A_\omega) \le \bar{\rho} \le 1$, and $\psi_k \sim \mathcal{N}(0,\Psi)$ is IID Gaussian noise. For stationarity, $\mbox{Cov}(\zeta_{0}) = A_{\omega}\mbox{Cov}(\zeta_{0})A_{\omega}^T + \Psi$. $\Delta u_k$ is independent of  other stochastic processes in the system. 
\begin{remark}
Here, $\bar{\rho}$, the maximum allowable spectral radius, is a design parameter for the defender. We observe a larger $\bar{\rho}$ improves expected detection performance. However, a larger $\bar{\rho}$ means a larger correlation between watermarks and this could facilitate the prediction of future watermarks if the attacker guesses an initial Gaussian input $\Delta u_k$. 
\end{remark}

\section{Attack Model} \label{attack}
In this section we describe a model of our adversary  in terms of knowledge, capabilities, and potential strategies.
\subsection{Attacker Capabilities}
Without loss of generality, we assume an attack begins at time $k = 0$. We make the following assumptions.
\begin{enumerate}
\item The attacker can modify all measurements $y_k$, $k \ge 0$. The falsified outputs at time $k$ are denoted by $y_k^v$.
\item The attacker inserts an input $B^a u_k^a$ into the system.
\item The attacker is unable to read the true control inputs $u_{k,c}$. As a result, he is unaware of the drop sequence $\{\eta_k\}$ and the Gaussian watermark $\{\Delta u_k\}$.
\end{enumerate}
The system under attack is given by
\begin{align}
x_{k+1} &= Ax_k + Bu_{k,c} + B^au_k^a + w_k, \\
\hat{x}_{k+1|k+1} &= (I-KC)(A\hat{x}_{k|k}+Bu_{k,c}) + Ky_{k+1}^v.
\end{align}
\begin{remark}
Attackers can inject $B^a u_k^a$ by appropriating the defender's actuators or inserting their own. The attacker could possibly modify inputs without being able to read them if the inputs are encrypted. Alternatively, the attacker can cause damage even if $B^au_k^a = 0$. For example, the attacker can destabilize the plant if $A$ is open loop unstable.
\end{remark}

\subsection{Attack Strategy} 
The attacker generates $y_k^v$ through a virtual system: 
\begin{align}
x_{k+1}^v &= Ax_k^v + \eta_k^v B (L_{m|b} \hat{x}_{k|k}^v+\Delta u_k^v) + w_k^v, \\
\hat{x}_{k+1|k+1}^v &= (I-KC)(A+\eta_k^v BL_{m|b})\hat{x}_{k|k}^v  + Ky_{k+1}^v\\ 
 &+ \eta_k^v(I-KC)B\Delta u_k^v , \nonumber \\
 y_k^v &= C x_k^v + v_k^v.
\end{align}
In the case of Watermark 1, $L_{m|b} = L_{(m)}$, $\eta_k^v$ follows a Markovian process \eqref{eq:markov} with parameters $\alpha$ and $\beta$ and $\Delta u_k^v \sim \mathcal{N}(0,\mathcal{Q})$ is an IID Gaussian process. In the case of Watermark 2, $L_{m|b} = L_{(b)}$, $\eta_k^v$ is an IID Bernoulli process with drop probability $p_d$ and $\Delta u_k^v$ is a stationary Gaussian process which satisfies \eqref{eq:stationary}. Additionally, $v_k^v \sim \mathcal{N}(0,R)$ and $w_k^v \sim \mathcal{N}(0,Q)$ are IID processes. Finally, we assume the stochastic processes $\{\eta_k^v,\Delta u_k^v, w_k^v, v_k^v\}$ are independent of the real system's stochastic parameters $\{\eta_k,\Delta u_k, w_k, v_k\}$.

The previous attack strategy can be generated (approximately) by a replay attack where the attacker records a long sequence of outputs $y_{-T':-T'+T}$ and, starting at time $0$, replaces $y_k$ with $y_k^v = y_{k-T'}$ for $0 \le k \le T$. Attackers who do not have precise knowledge of the model may engage in replay attacks, which only require access to the outputs \cite{Mo2009R,Chabukswar2013,Mo2014}. Alternatively, this attack strategy can be constructed by an adversary who is familiar with the model, for instance a malicious insider. In this case, the attacker simulates a virtual copy of the system dynamics to fool a bad data detector. It was previously shown \cite{Mo2014} that if $p_d = 0$ and there is no Gaussian watermark, the given strategies are asymptotically stealthy when $\mathcal{A} \triangleq (A+BL_{(b)})(I-KC)$ is Schur stable.

A model aware attacker could also potentially pursue an additive attack, for instance a false data injection attack \cite{moscs10security} or a zero dynamics attack \cite{PasqualettiJournal, teixeira2012}. In these attacks, the adversary injects an additive bias into the system which preserves the watermark and allows the attacker to remain stealthy. However, there are scenarios where additive attacks on sensor measurements are not feasible. As an example, suppose the defender uses public key cryptography, where a public key is used to encrypt the measurements while a private key is used to decrypt the associated cipher text. An attacker could send his own virtual measurements encrypted with the public key. However, such an attack could not leverage information in the true measurement as that would require access to the defender's private key to learn $y_k$. In this case, additive attacks constructed by replacing a true output packet with a virtual packet would be infeasible. By assumption, an additive networked-based attack on the defender's control input is also impossible because the adversary is unable to read the defender's input. 

We argue that alternative attack strategies which manipulate all sensors $y_k$ in a setting with public key cryptography also fail due to the fact that the resulting attack sequence $\{y_k^v\}$ is independent of the watermarks $\{\Delta u_k, \eta_k\}$. Specifically, an attacker who is unable to read the inputs or outputs will have no information about the watermarks. As a result, the outputs he can construct will fail to fool the correlation detector, which we propose in the next section.

\section{A Correlation Detector} \label{detector}
We consider a correlation detector, proposed in \cite{Chabukswar2013}. The defender computes a virtual output $y_k'$, which explicitly characterizes the effect of watermarks on $y_k$.
\begin{align}
&x_{k+1}' = Ax_k' + \eta_k B (L_{m|b} \hat{x}_{k|k}'+\Delta u_k),~~~y_k' = C x_k', \label{eq:virst} \\
&\hat{x}_{k+1|k+1}' = (I-KC)(A+\eta_k BL_{m|b})\hat{x}_{k|k}' + Ky_{k+1}' \label{eq:virstest} \\ 
&~~~~~~~~~~~+ \eta_k (I-KC)B \Delta u_k, \nonumber 
\end{align}
where with some abuse of notation  $x_{-\infty}' = 0, \hat{x}_{-\infty|-\infty}' = 0$. We can simplify \eqref{eq:virst} and \eqref{eq:virstest} to obtain
\begin{equation}
x_{k+1}' = (A + \eta_k BL_{m|b}) x_k' + \eta_k B \Delta u_k,~~~y_k' = C x_k' \label{eq:virsimple}.
\end{equation}
This virtual process created by the defender is driven entirely by the sequence of Bernoulli-Gaussian watermarks $\{\Delta u_k, \eta_k\}$. Thus, if we were to multiply the true outputs $y_k$ with the defender's virtual outputs $y_k'$ we would expect a positive correlation. However, if an attacker introduces measurements $y_k^v$, which are driven by an independent sequence of watermarks, the expected correlation drops to 0. This motivates consideration of the detection statistic $y_k^T y_k'$, where a large statistic is indicative of normal behavior while a small statistic indicates malicious behavior. Observe due to the random real time selection of watermarks, $\|y_k'\|_2$ may be close to 0, impacting detector performance since the correlation will likely also approach 0 even under normal operation. As a result, we propose an event triggered detector:
\begin{align}
\mbox{\bf{If} }  &\|y_k'\|_2^2 \ge \mu \nonumber  ~~~~~~~~~~ \mbox{\textit{Perform Detection}} \\
 &\kappa = \kappa+1, ~~t_{\kappa} = k \nonumber \\
&\sum_{j = \kappa - \mathcal{W} + 1}^{\kappa} g_j \overset{\mathcal{H}_0}{\underset{\mathcal{H}_1}{\gtrless}}  \tau ,~~~g_\kappa = y_{t_{\kappa}}^Ty_{t_{\kappa}}'.
\end{align}
The null hypothesis $\mathcal{H}_0$ is that the system is operating without malicious behavior while the alternative hypothesis $\mathcal{H}_1$ is that the system is under attack. $\mathcal{W}$ is the size of the detector's window. A detection event is triggered if $\|y_k'\|_2^2$ is greater than some user defined threshold $\mu$, preventing false alarms from being raised when $y_k'$ is small, while sacrificing time to detection. This tradeoff can be addressed by tuning $\mu$. Note that $\kappa$ corresponds to the time index of the event triggered correlation detector and increases at instants when a new detection statistic is computed. Identifying attacks on an individual sensor $i$ can be done by focusing on the correlation between individual measurements. An appropriate statistic $g_\kappa^i$ would be $y_{t_\kappa}^i{y_{t_\kappa}^i}'$ where $y_{t_\kappa}^i$ is the $i$th entry of $y_{t_\kappa}$.
\begin{remark} A detector with an adaptive threshold could address issues of small $y_k'$. However, such a detector is more prone to misses, mistaking an attack for noise. Incorporation and analysis of such a detector is left for future work.
\end{remark}

\begin{remark}
An adversary that can not read $\{u_k\},\{y_k\}$ can not take advantage of instances when detection does not occur, because such instances are entirely dependent on the realization of previous watermarks. An attacker who is forced to act independently of the real time watermarking sequence cannot determine if a detection has been triggered.
\end{remark}

We now verify that the expected correlation is $0$, if the outputs $y_k^v$ are generated independently of the watermarks.
\begin{theorem}
If $y_k^v$ and $\{\Delta u_k, \eta_k\}$ are independent, then 
\begin{equation*}
\mathbb{E}\left[{y_k^v}^T y_k' \Big| ~ \|y_k'\|_2^2 \ge \mu \right] = 0.
\end{equation*}
\begin{proof}
Observe that $y_k'$ can be written as a linear function of the Gaussian watermarks $\Delta u_k$ so that 
\begin{equation}
y_k' = \sum_{j = -\infty}^{k-1} G_j(\eta_{j:k-1}) \Delta u_j,
\end{equation}
where $G_j$ is some linear gain, determined by the sequence of Bernoulli drops $\eta_{j:k-1}$. Thus, we have
\begin{align*}
&\mathbb{E}[{y_k^v}^T y_k'] = \mathbb{E}\left[ {y_k^v}^T \sum_{j = -\infty}^{k-1} G_j(\eta_{j:k-1}) \Delta u_j \Bigg| ~ \|y_k'\|_2^2 \ge \mu \right] \\
                           &= \sum_{j = -\infty}^{k-1} \mathbb{E} \left[{y_k^v}\right]^T \mathbb{E}\left[ G_j(\eta_{j:k-1}) \Delta u_j \Big| ~ \|y_k'\|_2^2 \ge \mu  \right]  = 0.                                             \end{align*}
\end{proof}
\end{theorem}

The proposed detector can often differentiate between faulty and malicious scenarios. During a fault, we expect to see the effect of the embedded watermarks in the output and it could be measured through correlation. Alternatively, residue based detectors such as a $\chi^2$  detector ($g_{\kappa} = -z_{t_\kappa}^T(CPC^T+R)^{-1}z_{t_\kappa}$), which measures the difference between measured and expected behavior, will likely raise an alarm during faulty behavior and malicious behavior. Both detectors can be used in tandem. A residue based detector can raise alarms in the case of faulty or malicious behavior, while a correlation detector can distinguish these events. In this article, we focus on the correlation detector.
 

\section{The First Watermark Design} \label{watermark1}
We consider the design of a watermark consisting of an IID Gaussian input and Markovian drops. This requires the evaluation of a detection and performance trade-off. We wish to maximize the correlation of  $y_k$ and $y_k'$ to distinguish the system under attack from normal operation. However, we also need to ensure the system meets an adequate level of performance as characterized by the cost $\bar{J}$, starting at $k = 0$. 
\begin{equation}
\bar{J} = \lim_{N \rightarrow \infty} \frac{1}{N} \mathbb{E} \left[ \sum_{k = 0}^{N-1} x_k^TWx_k + u_{k,c}^T U u_{k,c} \right]
\end{equation}

As such, we design the parameters $\alpha, \beta, \mathcal{Q}$ by solving the following optimization problem
\begin{equation}
\begin{aligned}
& \underset{\alpha, \beta, \mathcal{Q}}{\text{maximize}}
& & \lim_{k \rightarrow \infty} \mathbb{E}[y_k^Ty_k'|\mathcal{H}_0] \\
& \text{subject to}
& & \bar{J} \leq \delta,  ~ 0 < \alpha, \beta \le 1.
\end{aligned} \label{eq:opt}
\end{equation}

To begin with, we use \cite[Theorem 3]{mo2013lqg} to analytically compute the cost $\bar{J}$ as follows. 
\begin{theorem}
Suppose $\alpha$ and $\beta$ are chosen so that the system has finite cost $J_{(m)}$ in the absence of a Gaussian watermark. The LQG cost $\bar{J}$ of the control system \eqref{eq:openloop} with IID Gaussian and Markovian watermark \eqref{eq:markovian input} is:
\begin{equation}
\bar{J} = J_{(m)}(\alpha,\beta) + \frac{\alpha}{\alpha+\beta} \mbox{tr}\left((B^TR_{(m)}B+U)\mathcal{Q}\right).
\end{equation}
\end{theorem}
\begin{proof}
Consider the cost to go in a finite horizon, $V_k(x_k) \triangleq \sum_{j = k}^N \mathbb{E}\left[ x_j^T W x_j + u_{j,c}^T U u_{j,c}  | \mathcal{F}_k \right]$, and let $u_{N,c} = 0$. Similar to, \cite{mo2013lqg}, it can be shown that 
\begin{equation}
V_k(x_k) = \begin{cases} \mathbb{E}[x_k^T S_k x_k | \mathcal{F}_k] + c_k & (\eta_{k-1} = 0) \\ \mathbb{E}[x_k^T R_k x_k | \mathcal{F}_k] + d_k & (\eta_{k-1} = 1) \end{cases},
\end{equation}
where $c_N = d_N = 0, ~ R_N, S_N = W, \bar{P} = P-KCP$ and
\begin{align}
F &= (A+BL_{(m)}) \nonumber, \\
R_k &= W + \beta A^T S_{k+1} A + \bar{\beta} F^TR_{k+1}F + \bar{\beta}L_{(m)}^T U L_{(m)} \nonumber, \\
S_k &= W + \bar{\alpha} A^T S_{k+1} A + \alpha F^TR_{k+1}F +  \alpha L_{(m)}^T U L_{(m)} \nonumber, \\
c_k &=  -\alpha \mbox{tr}((F^TR_{k+1}F-A^TR_{k+1}A + L_{(m)}^T U L_{(m)})(\bar{P})) \nonumber  \\
    &+ \alpha[\mbox{tr}(R_{k+1}Q) + d_{k+1} + \mbox{tr}((B^TR_{k+1}B+U)\mathcal{Q})] \nonumber \\
    &+ \bar{\alpha}[\mbox{tr}(S_{k+1}Q) + c_{k+1}], \label{eq:ck}  \\
d_k &=   -\bar{\beta} \mbox{tr}((F^TR_{k+1}F-A^TR_{k+1}A + L_{(m)}^T U L_{(m)})(\bar{P})) \nonumber  \\
    &+ \bar{\beta}[\mbox{tr}(R_{k+1}Q) + d_{k+1} + \mbox{tr}((B^TR_{k+1}B+U)\mathcal{Q})] \nonumber \\
    &+  \beta[\mbox{tr}(S_{k+1}Q) + c_{k+1}]. \label{eq:dk}
\end{align}
Let $\bar{J}_N = \mathbb{E} \left[ \sum_{k = 0}^N x_k^TWx_k + u_{k,c}^T U u_{k,c} \right] = \mathbb{E}[V_0(x_0)]$. We find that
\begin{align*}
\bar{J}_N &= P(\eta_{-1} = 0) \left(\mathbb{E}[x_{0}^T S_{0} x_{0}|\eta_{-1}=0] + c_{0}\right) \\ 
& + P(\eta_{-1} = 1) \left(\mathbb{E}[x_{0}^T R_{0} x_{0}|\eta_{-1}=1] + d_{0}\right).
\end{align*}
Leveraging the fact that $\{\eta_k\}$ is stationary with $P(\eta_k = 0) = \frac{\beta}{\alpha + \beta}$ as well as \eqref{eq:ck} and \eqref{eq:dk}, we obtain
\begin{align*}
\bar{J}_N &= \frac{1}{\alpha+\beta} \sum_{k = 0}^{N-1} \bigg(-\alpha \mbox{tr}((F^TR_{k+1}F-A^TR_{k+1}A \nonumber \\ &+ L_{(m)}^T U L_{(m)})(\bar{P})) + \mbox{tr}((\beta S_{k+1}+\alpha R_{k+1})Q) \\ &+ \alpha \mbox{tr}((B^T R_{k+1}B+U)\mathcal{Q}) \bigg) \\ &+ \frac{\beta \mathbb{E}[x_{0}^T S_{0} x_{0}|\eta_{-1}^0] + \alpha \mathbb{E}[x_{0}^T R_{0} x_{0}|\eta_{-1}^1]}{\alpha + \beta}, \nonumber
\end{align*} 
where $\eta_{-1}^j$ refers to the condition $\eta_{-1}=j$. It can be shown (in a similar manner to the proof of Theorem \ref{thm:correlation}) that the last term is bounded. Note $\bar{J} = \lim_{N \rightarrow \infty} \frac{1}{N} \bar{J}_{N-1}$. Moreover, from \cite{mo2013lqg}[Theorem 3, Lemma 4],  $\{S_k\}, \{R_k\}$ converge to $S_{(m)},R_{(m)}$, respectively. This proves the desired result.
\end{proof}
We now compute the expected correlation without attacks.
\begin{theorem} \label{thm:correlation}
Suppose $\alpha$ and $\beta$ are chosen so the resulting system has finite cost $J_{(m)}$ \cite{mo2013lqg}[Theorem 3] in the absence of a Gaussian watermark. Then, for the control system \eqref{eq:openloop} with IID Gaussian and Markovian watermark \eqref{eq:markovian input}, we have 
\begin{equation}
\lim_{k \rightarrow \infty} \mathbb{E}[y_k^Ty_k'|\mathcal{H}_0] = \frac{\mbox{tr}(C(\alpha X_1 + \beta X_0)C^T)}{\alpha + \beta},
\end{equation} 
where
\begin{align}
X_0 &= A(\bar{\alpha}X_0+\alpha X_1) A^T, \label{eq:fixedX0} \\
X_1 &= (A+BL_{(m)})(\beta X_0 + \bar{\beta}X_1)(A+BL_{(m)})^T + B\mathcal{Q}B^T  \nonumber
\end{align}
\end{theorem}
\begin{proof}
We begin with the Lemma below.
\begin{lemma} \label{lem:L0}
$\forall~M \in \mathbb{R}^{2n \times n}$, $\lim_{k \rightarrow \infty} \mathcal{L}_0^k(M) = 0$ where,
\begin{equation*}
\mathcal{L}_0 \begin{pmatrix} X \\ Y \end{pmatrix} = \begin{bmatrix} A(\bar{\alpha}{X} + {\alpha}{Y})A^T \\  (A+BL_{(m)})(\beta X  + \bar{\beta}Y )(A+BL_{(m)})^T \end{bmatrix}.
\end{equation*}
\end{lemma}
The proof is in the appendix. The closed loop  dynamics are
\begin{align*}
x_{k+1} &= (A+\eta_k BL_{(m)}) x_k - \eta_k BL_{(m)} e_k + w_k + \eta_kB\Delta u_k \\
e_{k+1} &= (A-KCA) e_k + (I-KC) w_k - Kv_{k+1},
\end{align*}
where $e_k = x_k - \hat{x}_{k|k}$.
From \eqref{eq:virsimple}, when $\eta_k = 1$. we obtain
\begin{align*}
&\mathbb{E}[x_{k+1}'x_{k+1}^T|\eta_k = 1]\\
& = (A+BL_{(m)})\mathbb{E}[x_k' x_k^T|\eta_k = 1](A+BL_{(m)})^T  - \\
&  (A+BL_{(m)})(\mathbb{E}[x_k' e_k^T | \eta_k = 1]L_{(m)}^TB^T - \mathbb{E}[x_k'w_k^T|\eta_k = 1]) \\
& + (A+BL_{(m)}) \mathbb{E}[x_k' \Delta u_k^T | \eta_k = 1] B^T  \\
& + B\mathbb{E}[\Delta u_k x_k^T | \eta_k = 1] (A+BL_{(m)})^T  \\
& + B\left(\mathbb{E}[\Delta u_k w_k^T| \eta_k = 1] +  \mathbb{E}[\Delta u_k \Delta u_k^T | \eta_k = 1] B^T \right) \\
& - B \mathbb{E}[\Delta u_k e_k^T | \eta_k = 1](BL_{(m)})^T,
\end{align*}
where we implicitly condition on $\mathcal{H}_0$. $x_k'$ is independent of $\Delta u_k, w_k, e_k$ and $\Delta u_k$ is independent of $x_k, w_k, e_k$. Thus,
\begin{align}
&\mathbb{E}[x_{k+1}'x_{k+1}^T|\eta_k = 1] \label{eq:notfixedx1} \\ 
& = (A+BL_{(m)})\mathbb{E}[x_k' x_k^T|\eta_k = 1](A+BL_{(m)})^T + B\mathcal{Q} B^T. \nonumber  
\end{align}
Next, since the Markov process is stationary and $x_k,x_k'$ and $\eta_k$ are conditionally independent given $\eta_{k-1}$, we observe
\begin{align}
&\mathbb{E}[x_k' x_k^T|\eta_k = 1]  \label{eq:split1} \\ 
&= P(\eta_{k-1} = 1|\eta_k = 1)\mathbb{E}[x_k' x_k^T|\eta_k = 1, \eta_{k-1} = 1] \nonumber \\
&+ P(\eta_{k-1} = 0|\eta_k = 1)\mathbb{E}[x_k' x_k^T|\eta_k = 1, \eta_{k-1} = 0], \nonumber \\
&= \bar{\beta}\mathbb{E}[x_k' x_k^T|\eta_{k-1} = 1] + \beta\mathbb{E}[x_k' x_k^T|\eta_{k-1} = 0]. \nonumber
\end{align}
It can be similarly shown that 
\begin{align}
&\mathbb{E}[x_{k+1}'x_{k+1}^T|\eta_k = 0] = A\mathbb{E}[x_k' x_k^T|\eta_k = 0]A^T.  \label{eq:notfixedx0} \\
&\mathbb{E}[x_k' x_k^T|\eta_k = 0]  \label{eq:split0}   \\
& = \alpha \mathbb{E}[x_k' x_k^T|\eta_{k-1} = 1] + \bar{\alpha}\mathbb{E}[x_k' x_k^T|\eta_{k-1} = 0]. \nonumber
\end{align}
Letting $X_{k,j} = \mathbb{E}[x_{k}'x_{k}^T|\eta_{k-1} = j]$ we have
\begin{equation}
\begin{pmatrix} X_{k+1,0} \\ X_{k+1,1} \end{pmatrix} = \mathcal{L}_0 \begin{pmatrix} X_{k,0} \\ X_{k,1} \end{pmatrix} + \begin{bmatrix} 0 \\ B\mathcal{Q}B^T \end{bmatrix}. \label{eq:L0}
\end{equation}

Since $\mathcal{L}_0$ is stable, $\lim_{k \rightarrow \infty} \mathbb{E}[x_{k}'x_{k}^T|\eta_{k-1} = 0]$ and $\lim_{k \rightarrow \infty} \mathbb{E}[x_k' x_k^T|\eta_{k-1} = 1]$ are  obtained by solving a fixed point equation which has a unique solution $X_0$ and $X_1$. \eqref{eq:fixedX0} immediately follows from \eqref{eq:L0}. Next, we find that
\begin{align}
\lim_{k \rightarrow \infty} \mathbb{E}[x_k' x_k^T] & = P(\eta_{k-1} = 1) X_1 + P(\eta_{k-1} = 0) X_0 , \label{eq:covariance} \\
& = \frac{\alpha X_1 + \beta X_0}{\alpha + \beta} \nonumber.
\end{align}
Finally, we observe that
\begin{equation}
\mathbb{E}[y_k^T y_k'] = \mbox{tr}\left(\mathbb{E}[(y_k'y_k^T)]\right) = \mbox{tr} \left(C\mathbb{E}[x_k'x_k^T]C^T\right). \label{eq:correlation} 
\end{equation}
\end{proof}

Thus, the watermark design problem \eqref{eq:opt} is given by 
\begin{equation}
\begin{aligned}
& \underset{\alpha, \beta, \mathcal{Q}}{\text{maximize}}
& & \frac{\mbox{tr}(C(\alpha X_1 + \beta X_0)C^T)}{\alpha + \beta} \\
& \text{subject to}
& & \begin{pmatrix} X_0 \\ X_1 \end{pmatrix} = \mathcal{L}_0\begin{pmatrix} X_0 \\ X_1 \end{pmatrix} + \begin{bmatrix} 0 \\ B\mathcal{Q}B^T \end{bmatrix}, \\
& & & J_{(m)}(\alpha,\beta) + \mbox{tr}((B^TR_{(m)}B+U)\mathcal{Q}) \leq \delta, \\
& & & 0 < \alpha, \beta \le 1.
\end{aligned} 
\end{equation}
For fixed $\alpha$ and $\beta$, the problem is an efficiently solvable semidefinite program. However, to optimize over $\alpha$ and $\beta$, we have to solve multiple instances of the problem over a finite 2 dimensional space. Ideally a designer will sample the space sufficiently. Note, not all $(\alpha,\beta)$ in $(0,~1] \times (0,~1]$ are feasible as some selections of $\alpha$ and $\beta$ lead to unbounded cost. Likewise, there may be naturally occurring drops which constrain $\alpha$ and $\beta$. For instance, if we add an artificial Markovian drop process on top of a naturally occurring IID drop process with drop probability $p_d$, we know that $\alpha \le (1-p_d), \bar{\beta} \le (1-p_d)$.

\begin{remark}
The optimal design of Watermark 1 requires solving multiple instances of a convex optimization problem with parameters varying over a bounded 2 dimensional space. This will also be true for Watermark 2. A formulation that considers a stationary Gaussian input with a Markovian drop process is nontrivial. Even if analysis can be performed, optimal design will likely require searching over 3 dimensions. This more complicated case is left for future work.
\end{remark}

\section{The Second Watermark Design} \label{watermark2}
We now investigate a watermark consisting of stationary Gaussian noise generated by a HMM \eqref{eq:stationary}  and an IID Bernoulli drop process at the control input with drop probability equal to $p_d$. Again, we design a watermark to address a performance and security trade-off. We wish to solve:
\begin{equation}
\begin{aligned}
& \underset{p_d, A_{\omega}, C_h, \Psi}{\text{maximize}}
& & \lim_{k \rightarrow \infty} \mathbb{E}[y_k^Ty_k'|\mathcal{H}_0] \\
& \text{subject to}
& & \bar{J} \leq \delta, ~~\rho(A_\omega) \le \bar{\rho},\\
& & & 0 \le p_d \le 1. 
\end{aligned} \label{eq:opt2}
\end{equation}

Rather than optimizing over the parameters of the HMM, we instead optimize over the autocovariance functions $\Gamma(d) \triangleq \mathbb{E}[\Delta u_k \Delta u_{k+d}^T]$. For tractable analysis we replace the constraint $\rho(A_\omega) \le \bar{\rho}$ with the following related assumption. \\
\textbf{Assumption 1:} Let $\Gamma(d)$ be an autocovariance function for a Gaussian process generated by an HMM $(A_\omega, C_h, \Psi)$. Then $(A_\omega, C_h, \Psi, \bar{\rho})$ is feasible only if $\tilde{\Gamma}(d) \triangleq \bar{\rho}^{-|d|} \Gamma(d)$ is a autocovariance function of a stationary Gaussian process. 

$\tilde{\Gamma}(d)$ can be potentially realized by an alternate HMM
\begin{equation}
\tilde{\zeta}_{k+1} = (A_{\omega}/\bar{\rho}) \tilde{\zeta}_{k} + \tilde{\psi}_k,~~\Delta \tilde{u}_k = C_h \tilde{\zeta}_{k}, \label{eq:tildeHMM}
\end{equation}
\begin{equation}
\mbox{Cov}(\tilde{\zeta}_0) = A_\omega \mbox{Cov}(\tilde{\zeta}_0)  A_\omega^T + \Psi, 
\end{equation}
\begin{equation}
  \tilde{\psi}_{k} \sim \mathcal{N}(0, \mbox{Cov}(\tilde{\zeta}_0) - A_\omega \mbox{Cov}(\tilde{\zeta}_0)  A_\omega^T/\bar{\rho}^2).
\end{equation}
Note, that if $\rho(A_\omega) > \bar{\rho}$, $\eqref{eq:tildeHMM}$ can not be a stationary process. This HMM can be realized if and only if  $\mbox{Cov}(\tilde{\zeta}_0) - A_\omega \mbox{Cov}(\tilde{\zeta}_0)  A_\omega^T/\bar{\rho}^2$ is positive semidefinite. Intuitively, if $\rho(A_\omega)$ is marginally less than $\bar{\rho}$, there is a larger chance that $\mbox{Cov}(\tilde{\zeta}_0) - A_\omega \mbox{Cov}(\tilde{\zeta}_0)  A_\omega^T/\bar{\rho}^2$ is positive semidefinite.
\begin{remark}
When $\bar{\rho} = 1$, Assumption 1, introduces no relaxation. In fact, the resulting formulation optimizes all stationary Gaussian processes in general. However, in the case $\bar{\rho} = 1$, we will prove that the resulting Gaussian process $\{\Delta u_k\}$ is entirely deterministic except for the initial watermark. A lower parameter $\bar{\rho}$ reduces average performance, but prevents an attacker who learns or guesses the current hidden state from adequately predicting future watermarks.
\end{remark}
We arrive at a relaxed formulation to \eqref{eq:opt2} below.
\begin{theorem} \label{thm:2ndwatermark}
Consider the control system \eqref{eq:openloop} with IID Bernoulli and stationary Gaussian watermark \eqref{eq:stationary}. Suppose $p_d$ is chosen so that the system has finite cost $J_{(b)}$ \cite{mo2013lqg}[Theorem 3] in the absence of a Gaussian watermark. An equivalent formulation to \eqref{eq:opt2} after replacing the constraint $\rho(A_\omega) \le \bar{\rho}$ with Assumption 1 is given by
\begin{equation}
\begin{aligned}
& \underset{\omega,H,p_d}{\text{maximize}}
& & \mbox{tr}(C F_2(\omega, H, p_d) C^T) \\
& \text{subject to}
& & J_{(b)}(p_d) + F_1(\omega, H, p_d) \leq \delta, \\
& & & 0 \leq p_d \leq 1,~~ 0 \leq \omega \leq 0.5, \\
& & & H \in \mathbb{C}^{p \times p},~~ H \succeq 0.
\end{aligned} \label{eq:opt3}
\end{equation}
where 
\begin{align}
&F_2(\omega,H,p_d) = 2\mbox{Re}\left(2\mbox{sym}\left[L_1(M_2 HB^T)\right]+ L_1(BHB^T)\right) \nonumber \\
&F_1(\omega,H,p_d) = \mbox{tr}(U\Theta) + \mbox{tr}((W+\bar{p}_dL_{(b)}^TUL_{(b)})F_2), \nonumber \\
& \Theta(\omega,H,p_d) = 2\mbox{Re}\left( 2\mbox{sym}\left[\bar{p}_d M_1 H \right] + \bar{p}_d H \right), \nonumber \\
&M_2 = \bar{p}_d \bar{\rho} s (A+BL_{(b)})\left[I - s\bar{\rho}(A+ \bar{p}_d BL_{(b)})\right]^{-1}B,  \nonumber \\
&M_1 = \bar{p}_d \bar{\rho} s L_{(b)} \left[ I - s \bar{\rho}(A+ \bar{p}_d BL_{(b)}) \right]^{-1}B, \nonumber \\
&L_1(X) = \bar{p}_d \left((A+BL_{(b)})L_1(X)(A+BL_{(b)})^T + X\right) \nonumber \\ 
&~~~~~~~~~+ p_d AL_1(X)A^T, \nonumber \\
& \mbox{sym}(X) = \frac{X + X^T}{2},~~~ s = \exp(2\pi j \omega),~~~\bar{p}_d = 1-p_d. \nonumber 
\end{align}
There is also an optimal solution $(H_*,\omega_*,{p_d}_*)$ such that $H_* = hh^H$ where $h^H$ denotes the conjugate transpose of $h \in \mathbb{C}^p$. Letting $\mbox{Re}$ and $\mbox{Im}$ be the real and imaginary parts of a matrix/vector, respectively, an optimal $A_{\omega}, C_h, \Psi$ is
\begin{align}
A_{\omega} = \bar{\rho} \begin{bmatrix} \cos(2 \pi \omega_*) & -\sin(2 \pi \omega_*) \\ \sin(2 \pi \omega_*) & \cos(2 \pi \omega_*) \end{bmatrix}, \nonumber \\
C_h = \sqrt{2} \begin{bmatrix} \mbox{Re}(h) & \mbox{Im}(h) \end{bmatrix}, ~~ \Psi = (1 - \bar{\rho}^2)I. \label{eq:optHMM} \
\end{align}
\end{theorem} 

The proof is similar in nature to the proof of Theorem 6 in \cite{Mo2014}. A sketch is found in the appendix. For fixed $p_d$ and $\omega$, the proposed problem is an efficiently solvable semidefinite program. To approximate a global maximum, we solve the problem repeatedly over the space $0 \le \omega \le 0.5$ and $0 \le p_d \le 1$. For sufficiently large $p_d$, the cost $\bar{J}$ becomes infinite in open loop unstable systems \cite{schenato2007foundations}, limiting the feasible space. We can account for natural packet drops in the system as before. For instance, if the input is dropped naturally with probability $p_d'$, we have $p_d' \le p_d \le 1$. 

\begin{remark}
An optimal watermark for a given $p_d \neq {p_d}_*$ may have better detection performance than the globally optimal watermark. Future work aims to use objective functions that better highlight the relative performance of watermarks.
\end{remark}
\begin{remark}
While packet drops at the sensor measurements are not modeled in this paper, our framework could be extended to address this behavior without significantly changing the formulations of the proposed optimization problems. The main effect of packet drops at the sensor side is a time varying Kalman gain. The objective function and increase in cost $\bar{J}$ due to the Gaussian portion of the watermark are not affected by time variations in the Kalman gain in both watermarking settings. Both $J_{(m)}$ and $J_{(b)}$ can be empirically evaluated for fixed $(\alpha,\beta)$ and $p_d$, respectively, to account for packet drops at the sensor measurements.
\end{remark}

\section{Simulations} 
\label{simulations}

In this section, we illustrate the performance of the proposed watermarking designs through extensive numerical results. We tested our watermark designs in various randomly generated systems and, unless otherwise stated, averaged results over 1500 trials. Replay attacks are considered.

In Fig. \ref{fig:rocplots2}, we utilize Watermark 1, which has a Markovian drop process defined by parameters $(\alpha,\beta)$ and an IID Gaussian watermark. The watermark is tested on a randomly generated open loop stable system with 5 states, 4 inputs, and 2 outputs. We plot  the receiver operating characteristic (ROC) curve for both the proposed correlation detector and a $\chi^2$ detector. The $\chi^2$ detector serves as a benchmark, having been previously used for attack detection \cite{Mo2009R,Chabukswar2013,miao2013stochastic,omur_smartgridcomm} in watermarked systems. The threshold $\mu$ is chosen to be a constant multiple of $\lim_{k \rightarrow \infty} E[y_k^Ty_k']$. The ROC curves are collected at multiple different costs $\Delta J=1.05 J^{*}$, $\Delta J=0.45 J^{*}$ and $\Delta J=0.15 J^{*}$. Here, $\Delta J$ represent the increase in the cost $\bar{J}$ relative to optimal cost $J^*$ without drops or a Gaussian watermark. We compare a system with drops $(\alpha = 0.69, \beta = 0.9)$ to a system without drops $(\alpha = 1,\beta = 0)$. The proposed detector outperforms the $\chi^2$ detector in all cases and packet drops improve the ROC curve for both detectors. The improvement appears to be higher for moderately valued $\Delta J$ before saturating. In Fig. \ref{fig:timeplots2}, we plot the expected time to detection for both detectors in a system with Watermark 1. The packet drop process introduces an additional delay in the time to detection though this additional time is less significant as $\Delta J$ is increased.

\begin{figure}[ht]
\centering
\subfloat[Correlation Detector]{%
\includegraphics[scale=0.21]{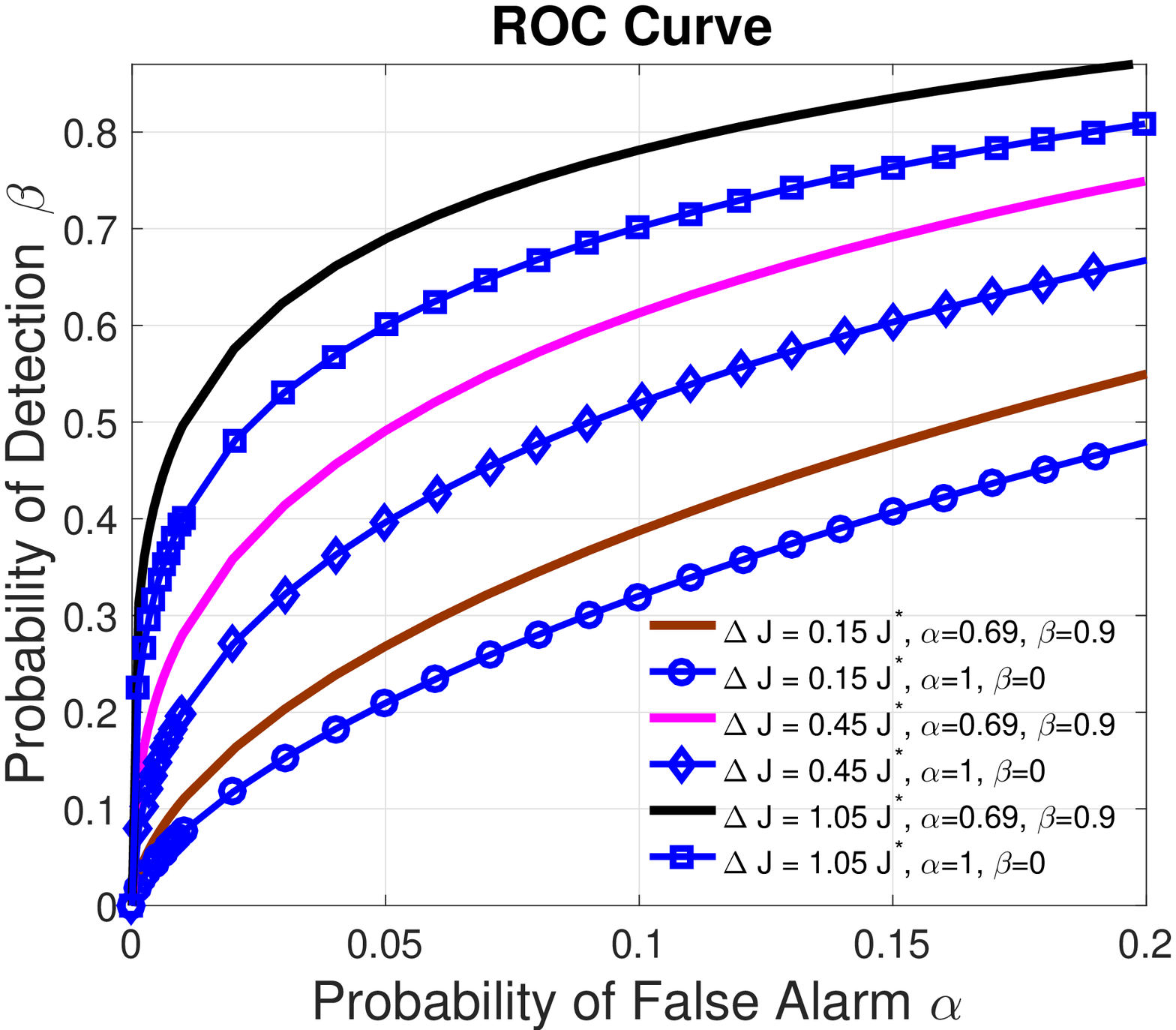}
\label{fig:ROC_3}}
\subfloat[$\chi^2$ Detector]{%
\includegraphics[scale=0.22]{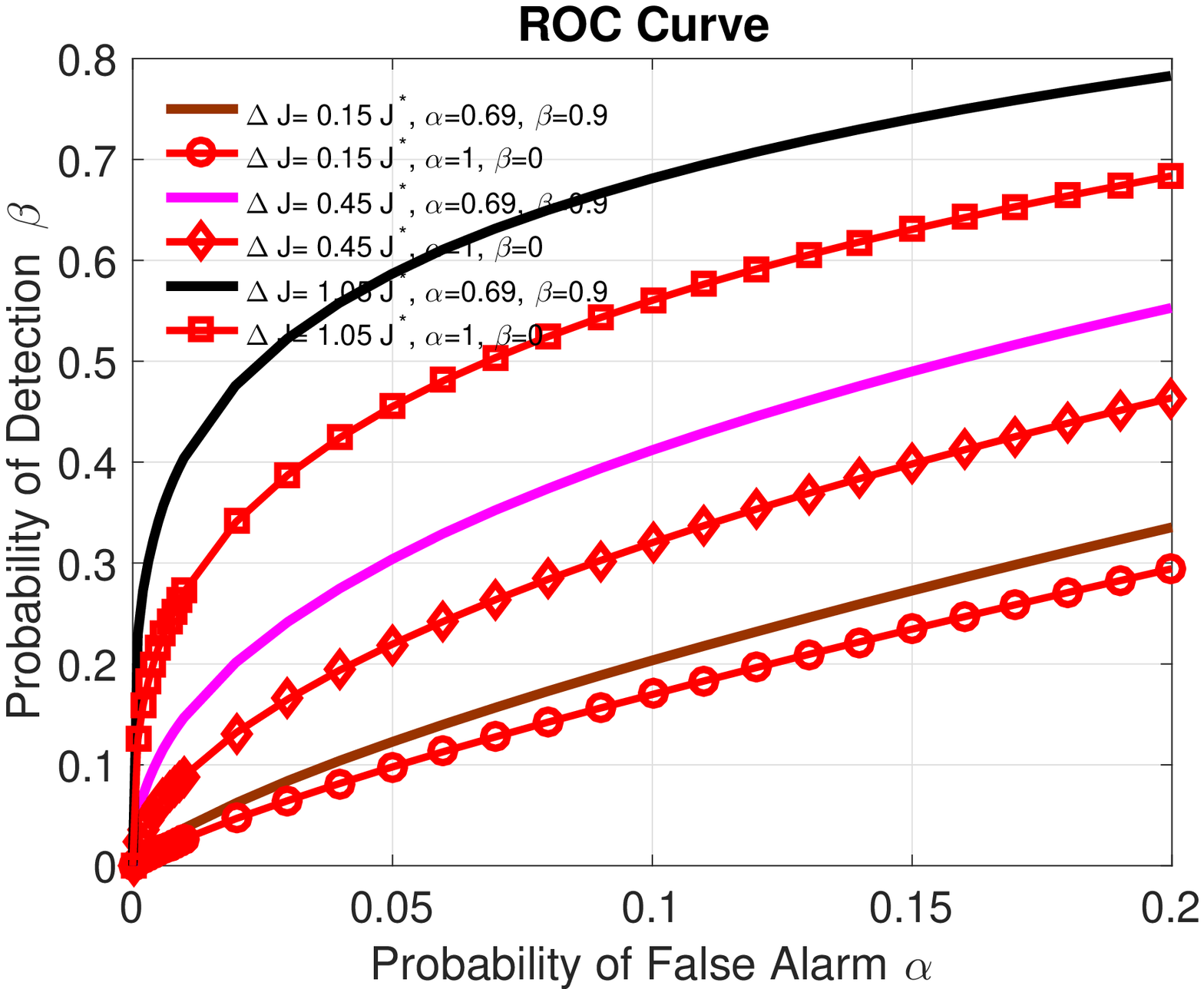}
\label{fig:ROC_4}} \\
\caption{Detection probability versus false alarm rate for $\chi^2$ and correlation detectors for a system using Watermark 1.}
\label{fig:rocplots2}
\end{figure}

\begin{figure}[ht]
\centering
\subfloat[Correlation Detector]{%
\includegraphics[scale=0.22]{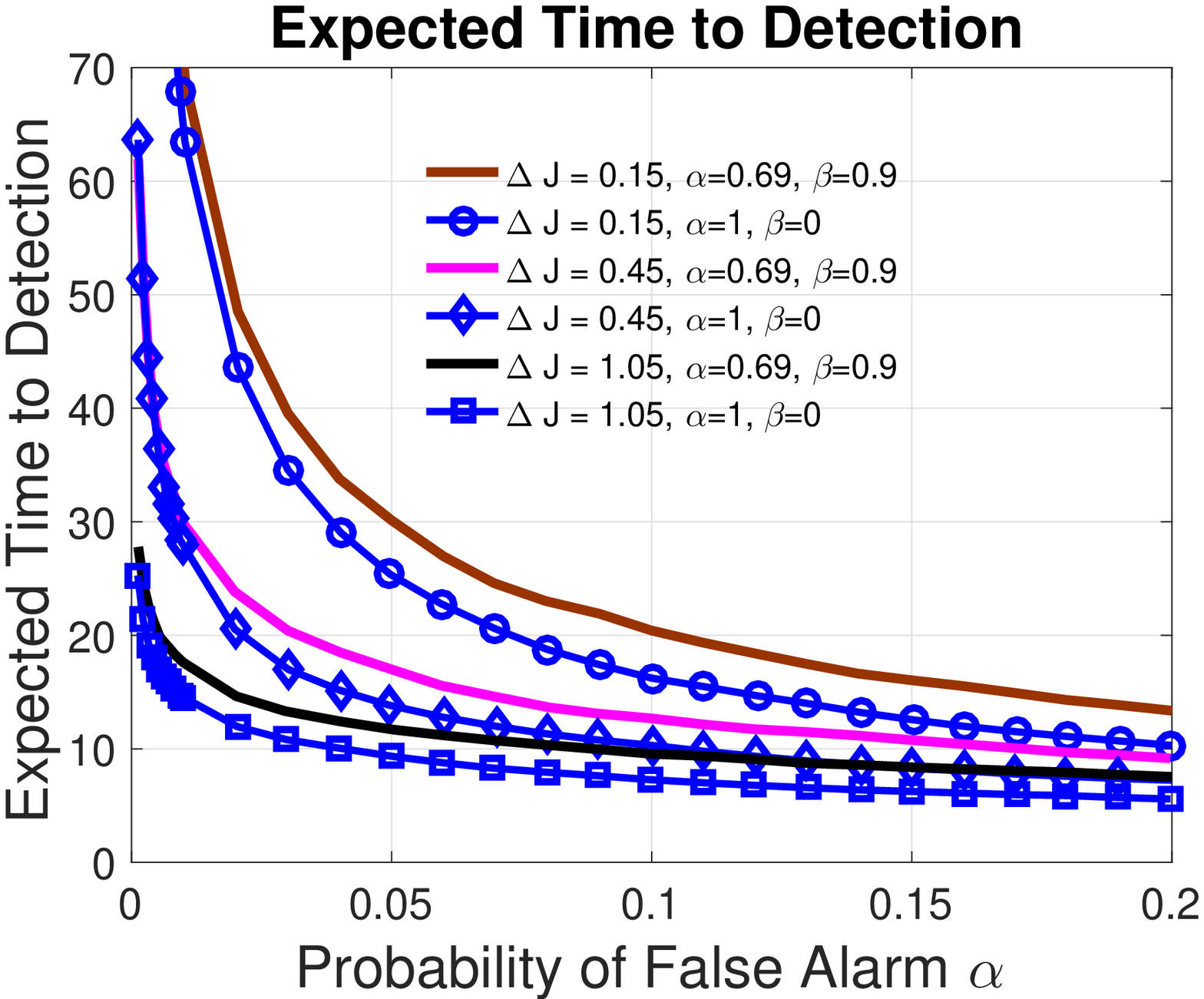}
\label{fig:time_3}}
\subfloat[$\chi^2$ Detector]{%
\includegraphics[scale=0.22]{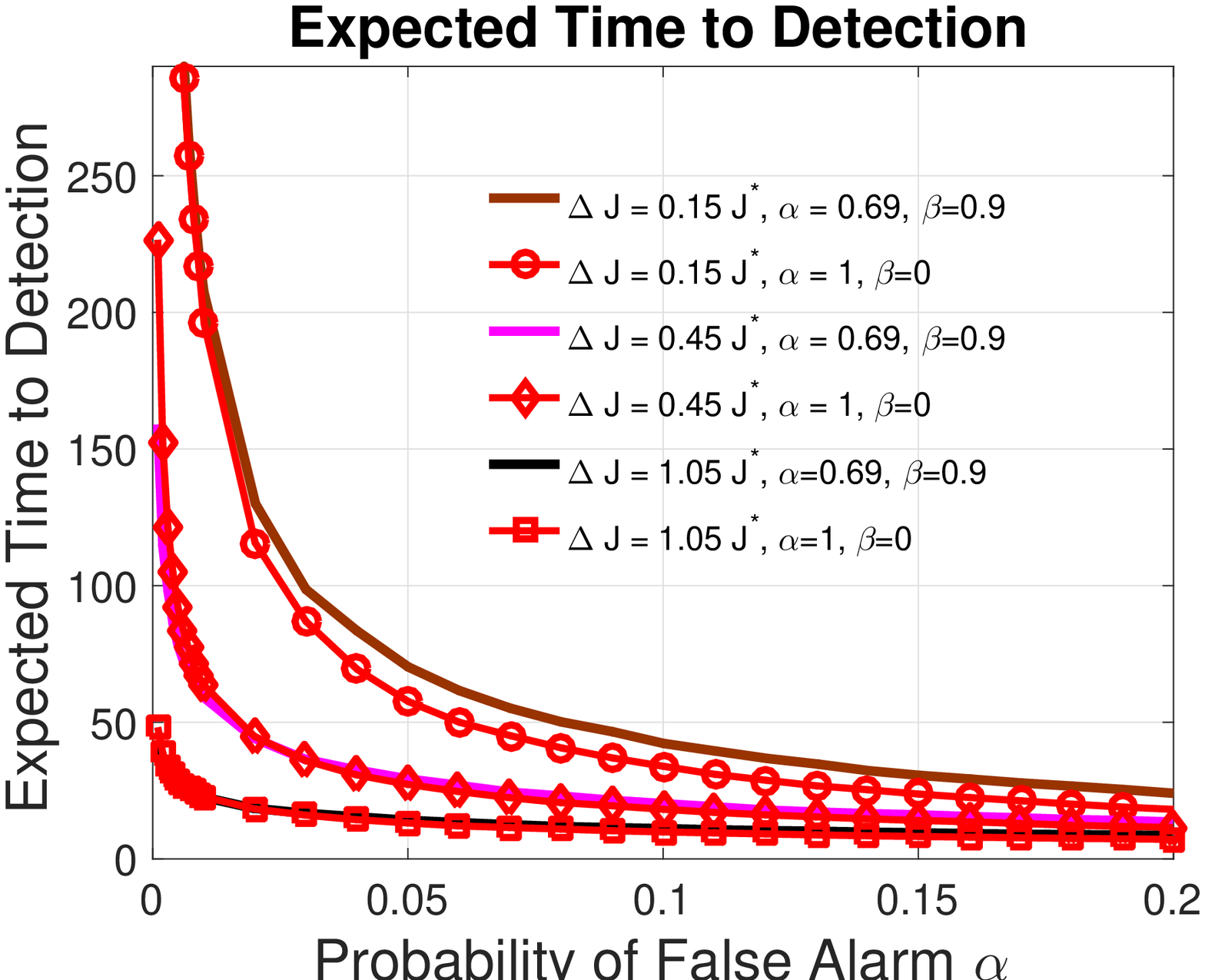}
\label{fig:time_4}} \\
\caption{Expected time to detection for $\chi^2$ and correlation detectors for a system using Watermark 1.}
\label{fig:timeplots2}
\end{figure}

In Fig. \ref{fig:rocplots}, we introduce Watermark 2, which has IID drops (with probability of drop $p_d$) and a stationary Gaussian watermark. The watermark is added to a randomly generated open loop stable system with 6 states, 5 inputs, and 5 outputs. We plot ROC curves generated by both the correlation detector and $\chi^2$ detector for a system with drops $(p_d = 0.6)$ and a system without drops $(p_d = 0)$, at various costs of control $\Delta J = 0.95 J^{*}$, $\Delta J = 0.45 J^{*}$ and $\Delta J = 0.15 J^{*}$. Time to detection plots are provided in Fig. \ref{fig:timeplots}. The results and patterns observed here are similar to the results seen in the system with Watermark 1.

\begin{figure}[ht]
\centering
\subfloat[Correlation Detector]{%
\includegraphics[scale=0.21]{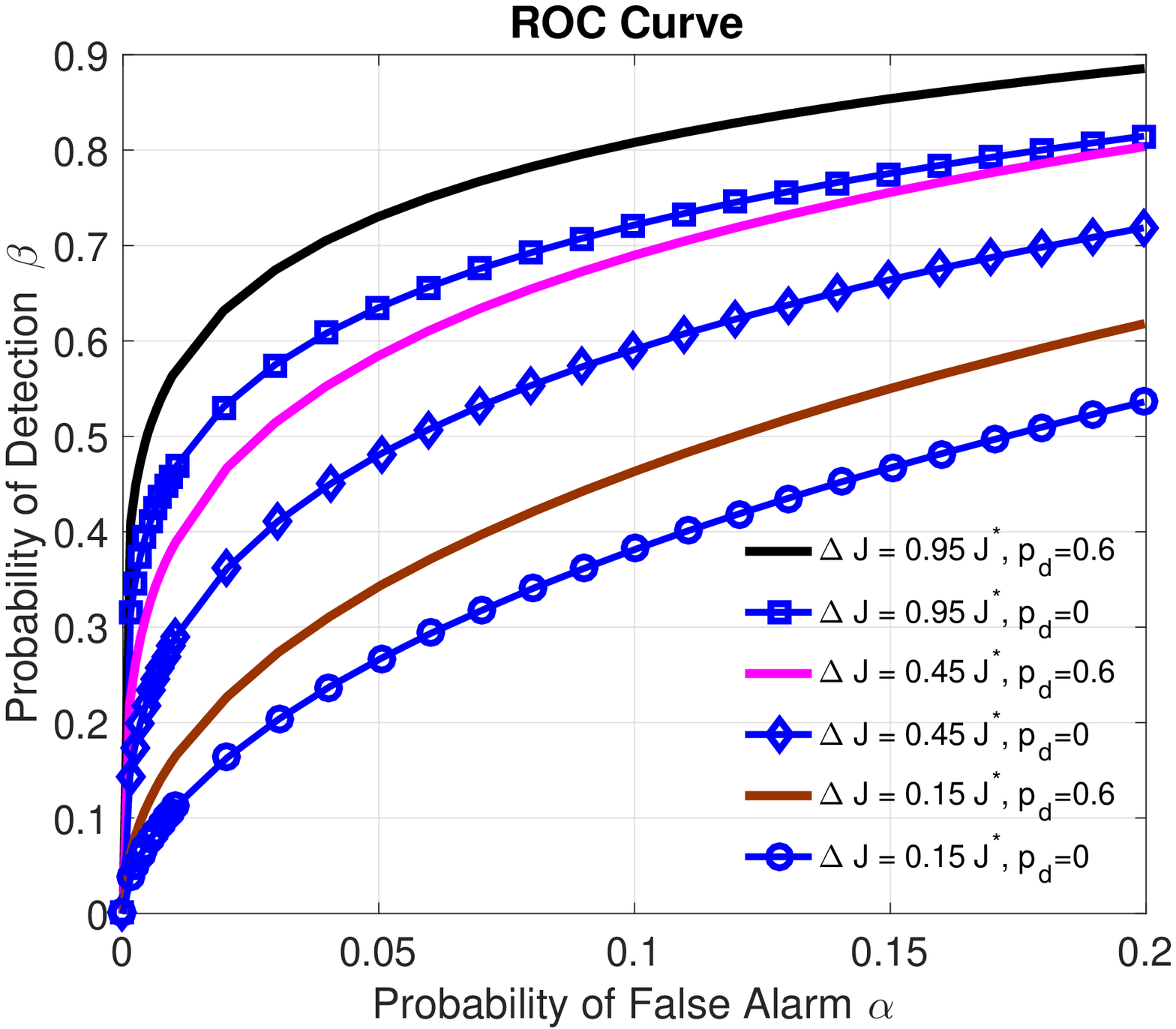}
\label{fig:ROC_1}}
\subfloat[$\chi^2$ Detector]{%
\includegraphics[scale=0.21]{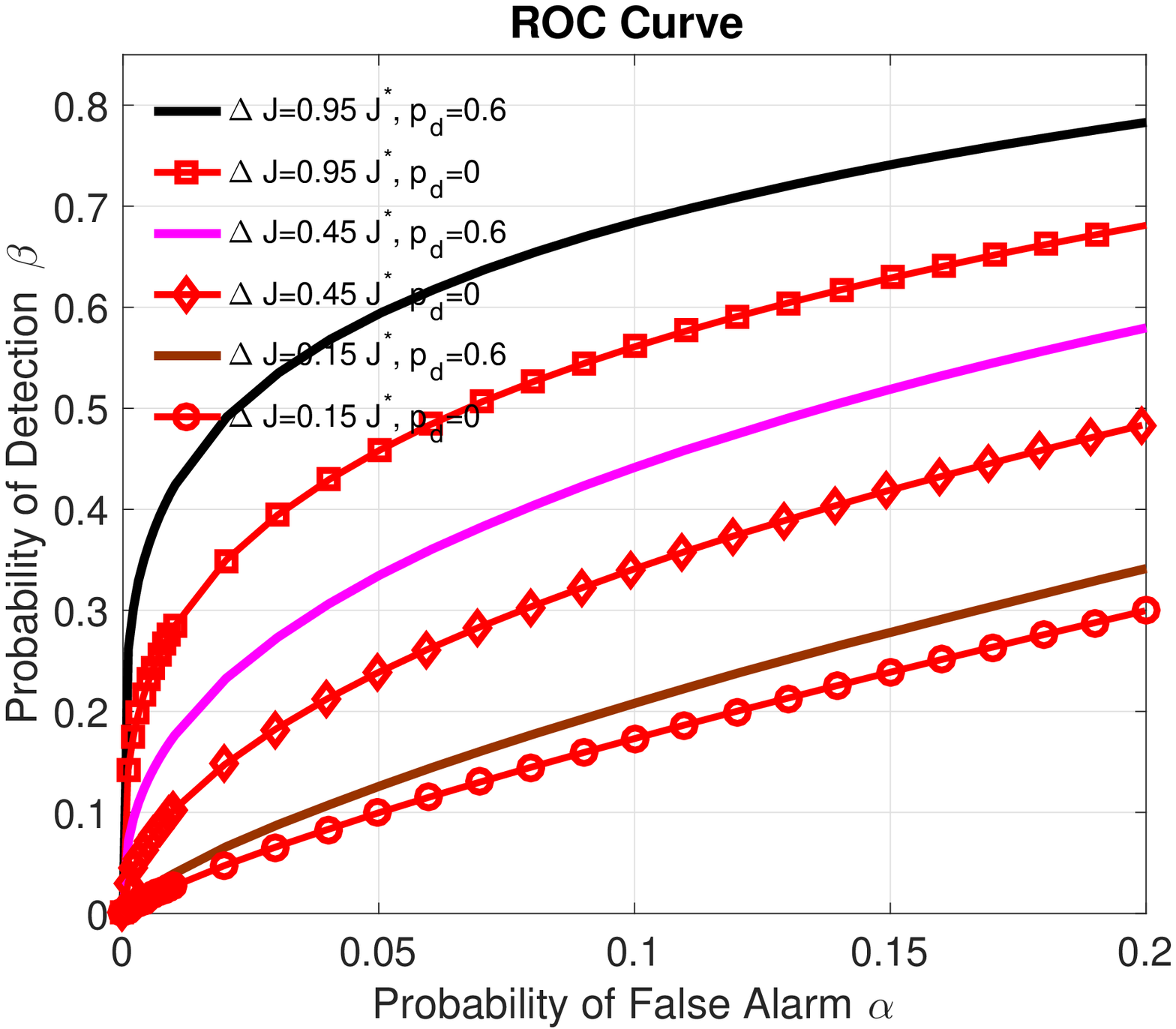}
\label{fig:ROC_2}} \\
\caption{Detection probability versus false alarm rate for $\chi^2$ and correlation detectors for a system using Watermark 2.}
\label{fig:rocplots}
\end{figure}

\begin{figure}[ht]
\centering
\subfloat[Correlation Detector]{%
\includegraphics[scale=0.22]{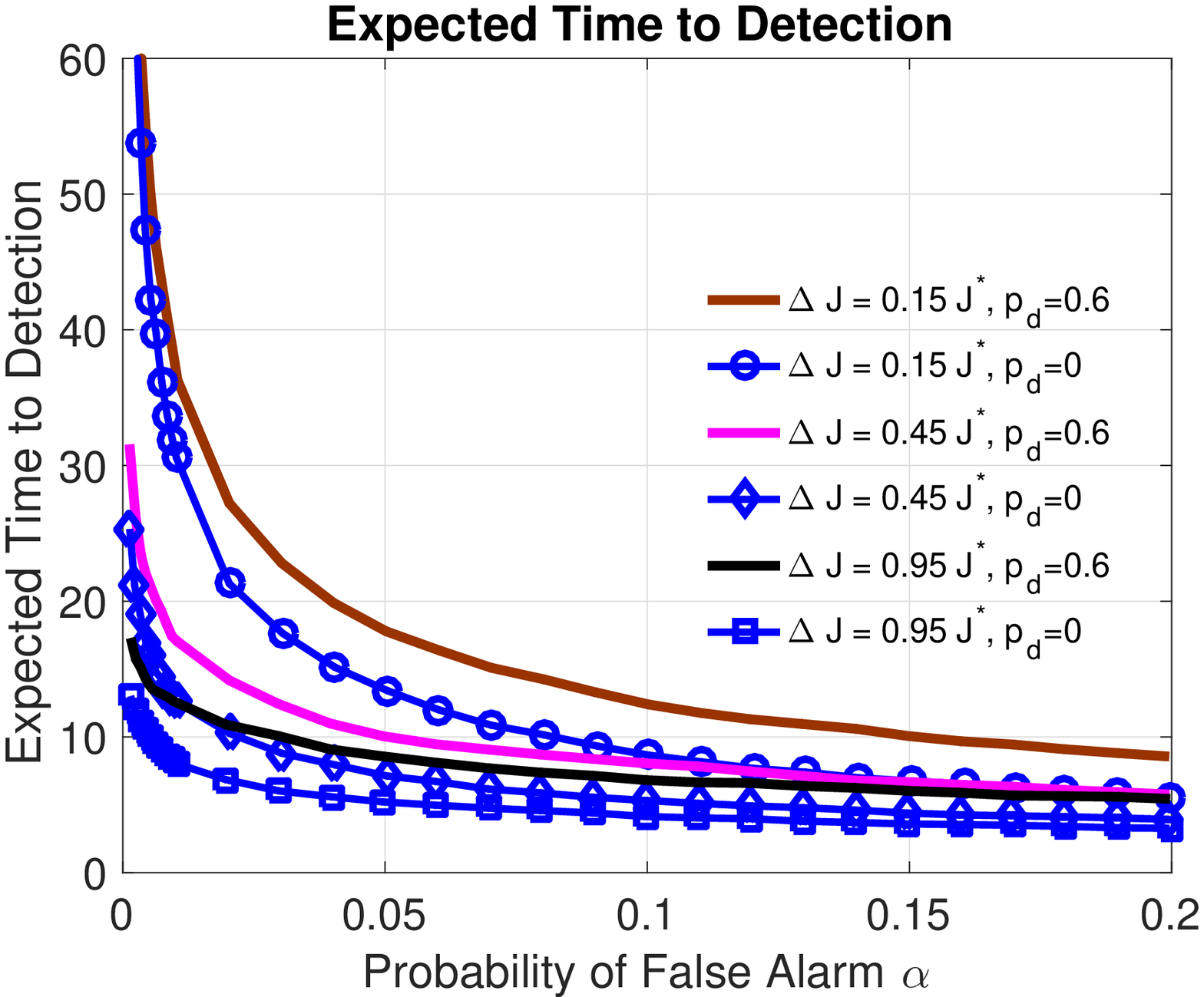}
\label{fig:time_1}}
\subfloat[$\chi^2$ Detector]{%
\includegraphics[scale=0.22]{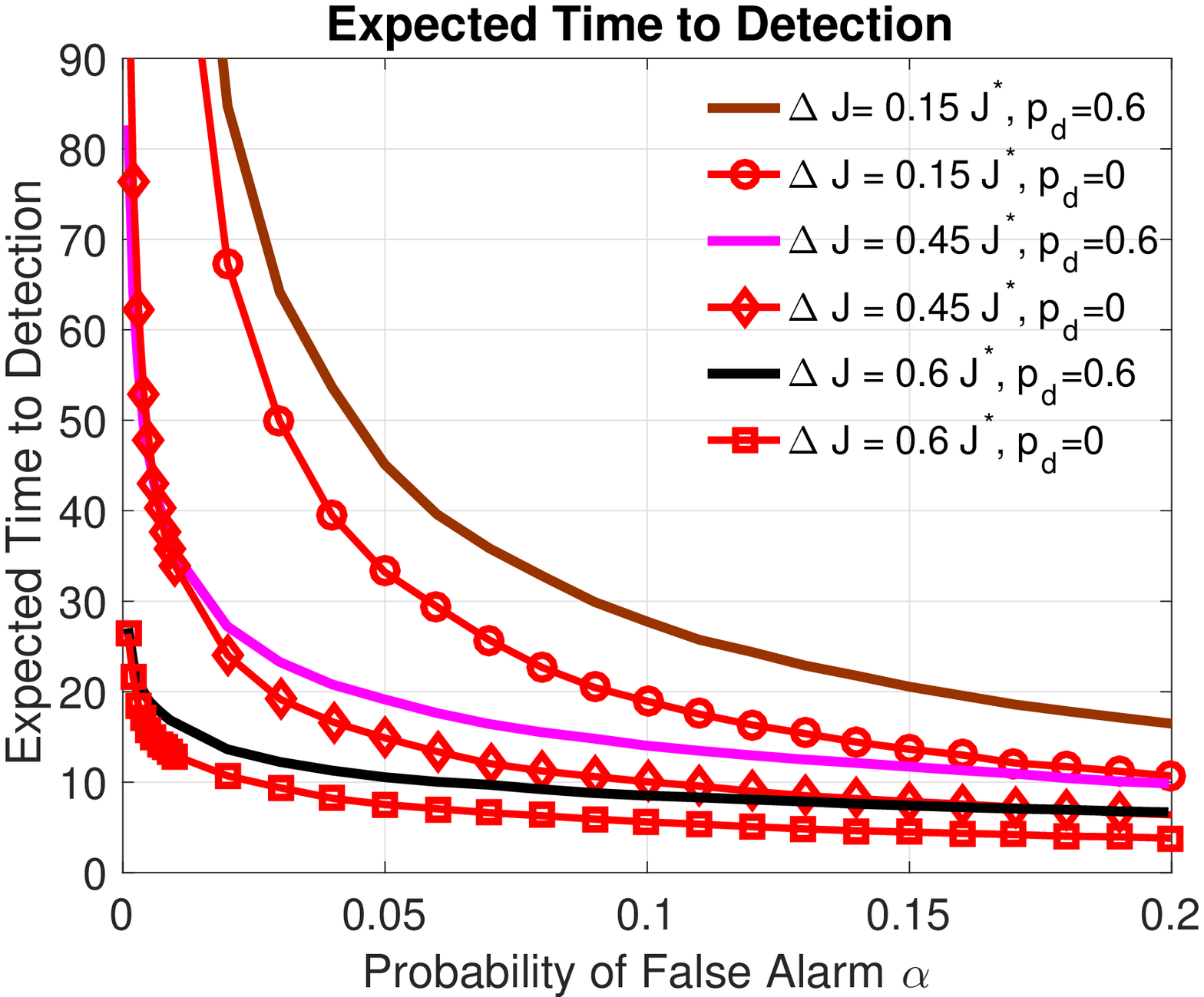}
\label{fig:time_2}} \\
\caption{Expected time to detection for $\chi^2$ and correlation detectors for a system using Watermark 2.}
\label{fig:timeplots}
\end{figure}

In Figs. \ref{fig:faultmarkov} and \ref{fig:faultbernoulli}, we plot $\chi^2$ detector and correlation detector statistics (averaged over 500 trials) during a  fault in the system. The fault introduced (at time $210$) is a constant additive bias added to a subset of sensors (i.e. due to disturbances/sensor drift). While the $\chi^2$ detector raises an alarm, the correlation detector does not since the watermark is preserved in the system. This motivates the use of both the correlation and $\chi^2$ detector to distinguish faults from attacks. If both detectors raise an alarm, indicating the watermark is absent in the outputs, we consider a likely attack scenario. If only the $\chi^2$ detector raises an alarm, we expect that the watermark is preserved while the dynamics are inconsistent with modeling. As such, we anticipate a fault.


%

\begin{figure}[ht]
\centering
\subfloat[Correlation Detector]{%
\includegraphics[scale=0.22]{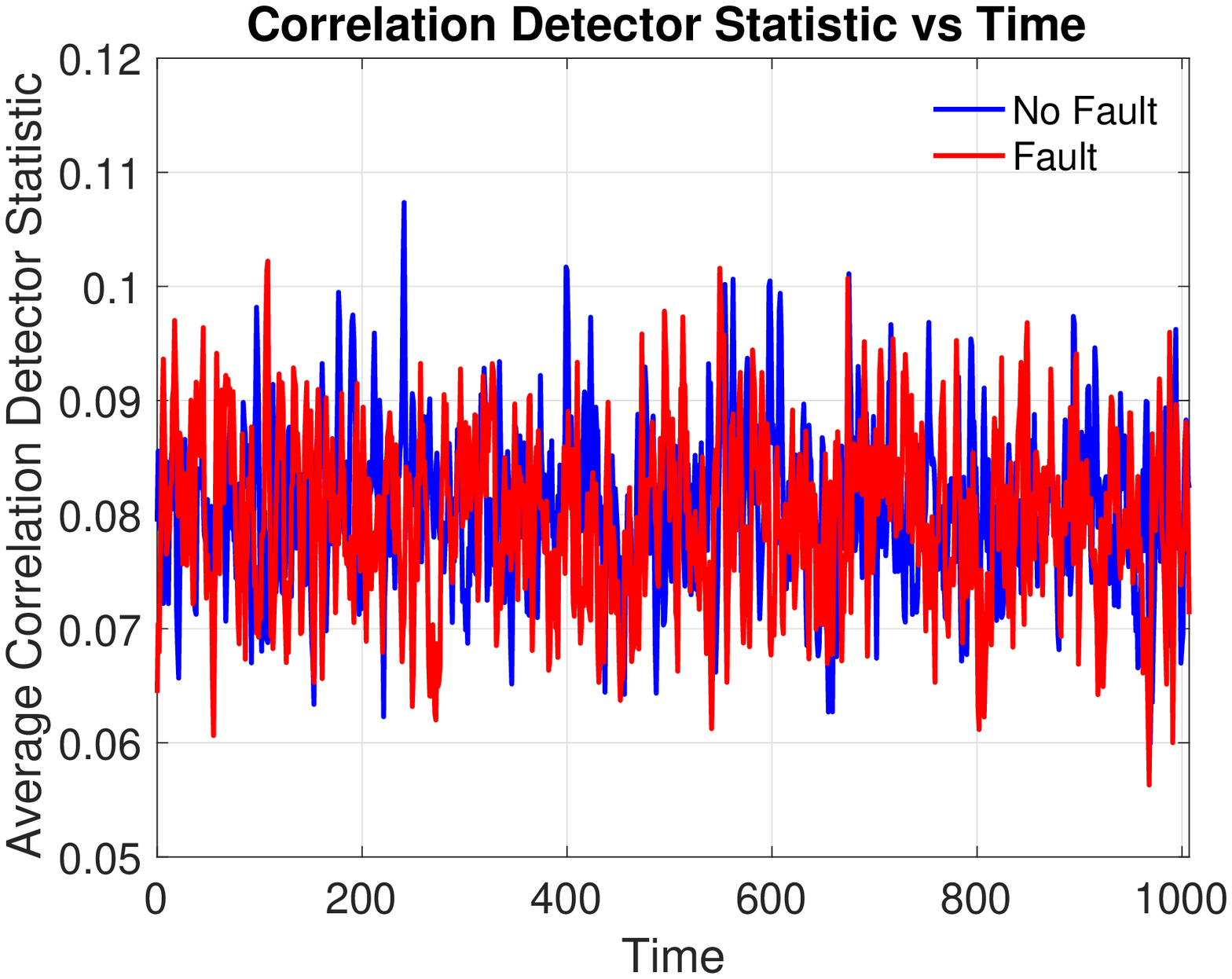}
\label{fig:fault_3}}
\subfloat[$\chi^2$ Detector]{%
\includegraphics[scale=0.22]{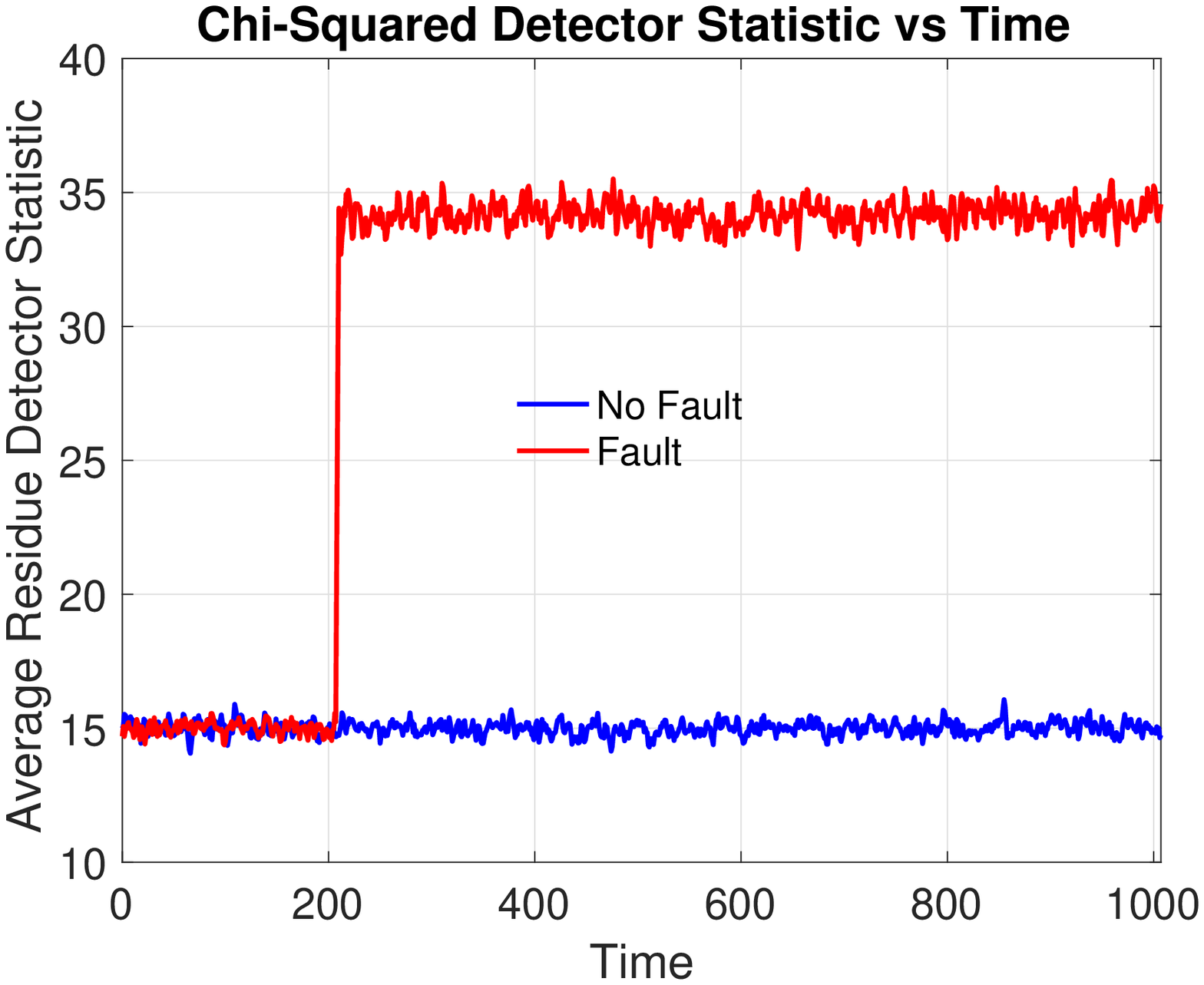}
\label{fig:fault_4}} \\
\caption{Average correlation detector and $\chi^2$ detector statistics under a fault at the sensor output for a system using Watermark 1.}
\label{fig:faultmarkov}
\end{figure}

\begin{figure}[ht]
\centering
\subfloat[Correlation Detector]{%
\includegraphics[scale=0.22]{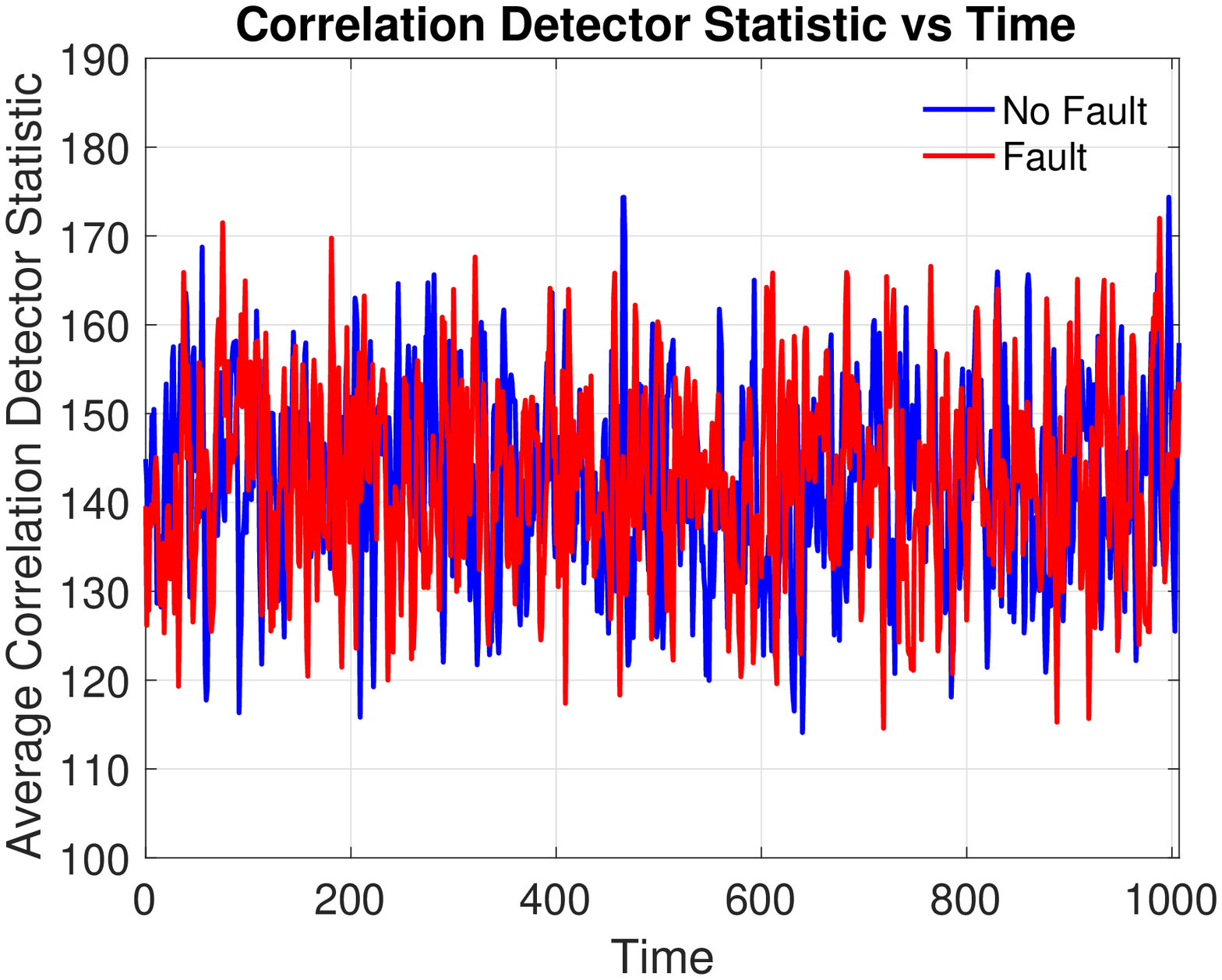}
\label{fig:fault_1}}
\subfloat[$\chi^2$ Detector]{%
\includegraphics[scale=0.22]{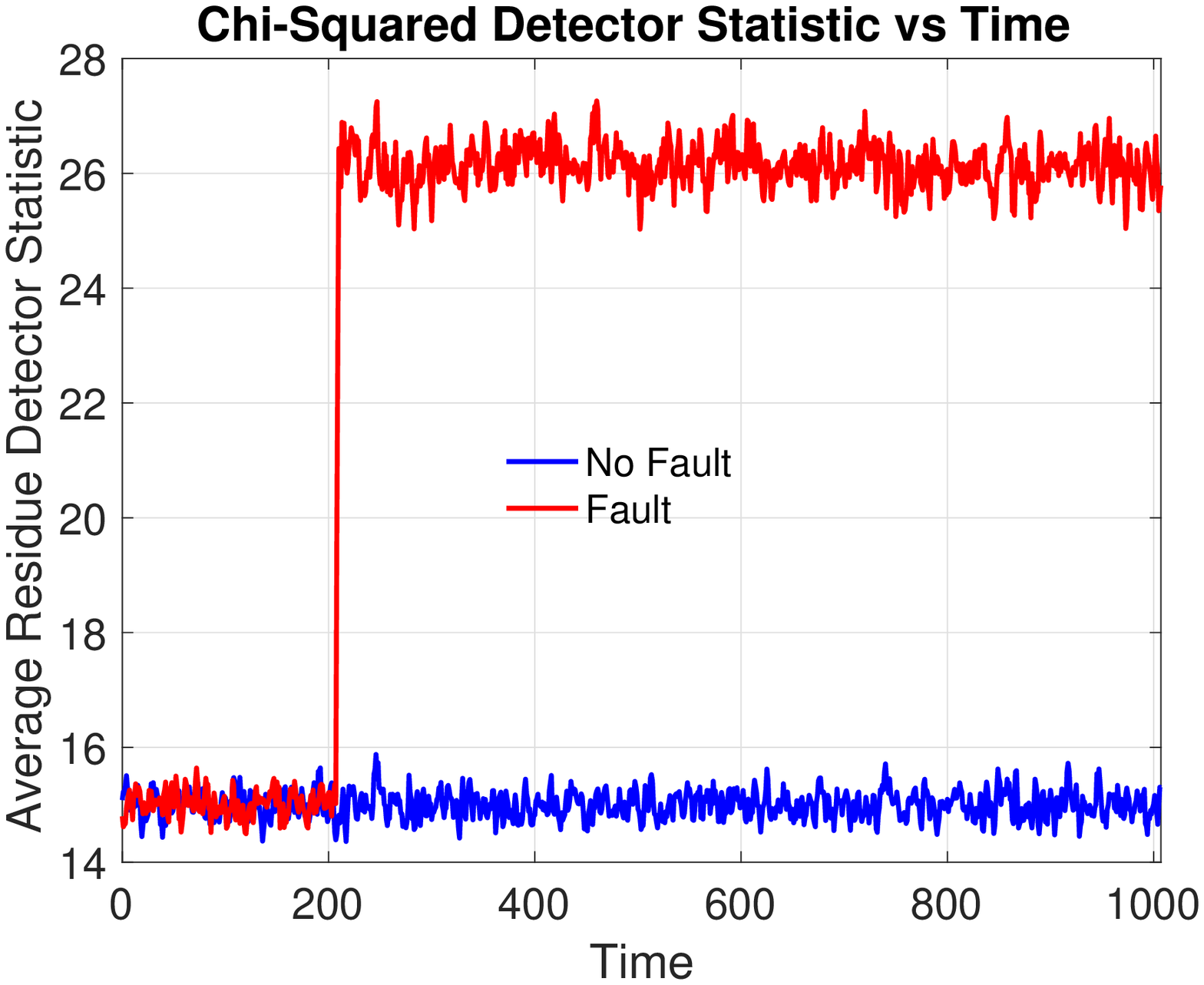}
\label{fig:fault_2}} \\
\caption{Average correlation detector and $\chi^2$ detector statistics under a fault at the sensor output for a system using Watermark 2.}
\label{fig:faultbernoulli}
\end{figure}


\section{Conclusion and Future Work} 
\label{conclusion}
In this paper, we showed how to incorporate Bernoulli packet drops at the control input in the design of physical watermarks. We argued that packet drops can be beneficial for detection and consequently considered the design of a joint Bernoulli-Gaussian watermark to detect integrity attacks. We proposed two main watermark designs in conjunction with a correlation detector and provided efficiently solvable optimization problems to address the trade-off between detection and control performances. In future work we aim to generalize our watermarking approach to allow us to drop either the entire control input or the Gaussian portion of the watermark. We also hope to conduct testing in real systems.

\bibliographystyle{IEEEtran}
\bibliography{identification_cps2} 

\begin{thebibliography}{10}
\providecommand{\url}[1]{#1}
\csname url@samestyle\endcsname
\providecommand{\newblock}{\relax}
\providecommand{\bibinfo}[2]{#2}
\providecommand{\BIBentrySTDinterwordspacing}{\spaceskip=0pt\relax}
\providecommand{\BIBentryALTinterwordstretchfactor}{4}
\providecommand{\BIBentryALTinterwordspacing}{\spaceskip=\fontdimen2\font plus
\BIBentryALTinterwordstretchfactor\fontdimen3\font minus
  \fontdimen4\font\relax}
\providecommand{\BIBforeignlanguage}[2]{{%
\expandafter\ifx\csname l@#1\endcsname\relax
\typeout{** WARNING: IEEEtran.bst: No hyphenation pattern has been}%
\typeout{** loaded for the language `#1'. Using the pattern for}%
\typeout{** the default language instead.}%
\else
\language=\csname l@#1\endcsname
\fi
#2}}
\providecommand{\BIBdecl}{\relax}
\BIBdecl

\bibitem{Cardenas:2008ke}
A.~A. C{\'a}rdenas, S.~Amin, and S.~S. Sastry, ``{Secure Control: Towards
  Survivable Cyber-Physical Systems},'' in \emph{Distributed Computing Systems
  Workshops, 2008. ICDCS '08. 28th International Conference on DOI -
  10.1109/ICDCS.Workshops.2008.40}.\hskip 1em plus 0.5em minus 0.4em\relax
  IEEE, 2008, pp. 495--500.

\bibitem{Chen2010}
T.~M. Chen, ``Stuxnet, the real start of cyber warfare? [editor's note],''
  \emph{IEEE Network}, vol.~24, no.~6, pp. 2--3, 2010.

\bibitem{Slay2008}
J.~Slay and M.~Miller, ``Lessons learned from the {M}aroochy water breach,'' in
  \emph{Critical Infrastructure Protection}.\hskip 1em plus 0.5em minus
  0.4em\relax Springer US, 2008, pp. 73--82.

\bibitem{liang20172015}
G.~Liang, S.~R. Weller, J.~Zhao, F.~Luo, and Z.~Y. Dong, ``The 2015 ukraine
  blackout: Implications for false data injection attacks,'' \emph{IEEE
  Transactions on Power Systems}, vol.~32, no.~4, pp. 3317--3318, 2017.

\bibitem{Mo2009R}
Y.~Mo and B.~Sinopoli, ``Secure control against replay attacks,'' in \emph{47th
  Annual Allerton Conference on Communication, Control, and Computing}, Sept
  2009, pp. 911--918.

\bibitem{bai2015security}
C.-Z. Bai, F.~Pasqualetti, and V.~Gupta, ``Security in stochastic control
  systems: Fundamental limitations and performance bounds,'' in \emph{American
  Control Conference (ACC), 2015}.\hskip 1em plus 0.5em minus 0.4em\relax IEEE,
  2015, pp. 195--200.

\bibitem{weerakkody2016informationflow}
S.~Weerakkody, B.~Sinopoli, S.~Kar, and A.~Datta, ``Information flow for
  security in control systems,'' in \emph{55th IEEE Conference on Decision and
  Control (CDC)}.\hskip 1em plus 0.5em minus 0.4em\relax IEEE, 2016, pp.
  5065--5072.

\bibitem{Chabukswar2013}
Y.~Mo, R.~Chabukswar, and B.~Sinopoli, ``Detecting integrity attacks on {SCADA}
  systems,'' \emph{IEEE Transactions on Control Systems Technology}, vol.~22,
  no.~4, pp. 1396--1407, 2014.

\bibitem{Mo2014}
Y.~Mo, S.~Weerakkody, and B.~Sinopoli, ``Physical authentication of control
  systems{:} designing watermarked control inputs to detect counterfeit sensor
  outputs,'' \emph{IEEE Control Systems Magazine}, vol.~35, no.~1, pp. 93 --
  109, 2015.

\bibitem{miao2013stochastic}
F.~Miao, M.~Pajic, and G.~J. Pappas, ``Stochastic game approach for replay
  attack detection,'' in \emph{52nd IEEE Conference on Decision and Control
  (CDC)}.\hskip 1em plus 0.5em minus 0.4em\relax IEEE, 2013, pp. 1854--1859.

\bibitem{satchidanandan2017dynamic}
B.~Satchidanandan and P.~Kumar, ``Dynamic watermarking: Active defense of
  networked cyber--physical systems,'' \emph{Proceedings of the IEEE}, vol.
  105, no.~2, pp. 219--240, 2017.

\bibitem{hespanhol2017dynamic}
P.~Hespanhol, M.~Porter, R.~Vasudevan, and A.~Aswani, ``Dynamic watermarking
  for general {LTI} systems,'' \emph{arXiv preprint arXiv:1703.07760}, 2017.

\bibitem{hosseini2016designing}
M.~Hosseini, T.~Tanaka, and V.~Gupta, ``Designing optimal watermark signal for
  a stealthy attacker,'' in \emph{2016 European Control Conference
  (ECC)}.\hskip 1em plus 0.5em minus 0.4em\relax IEEE, 2016, pp. 2258--2262.

\bibitem{rubio2017use}
J.~Rubio-Hernan, L.~De~Cicco, and J.~Garcia-Alfaro, ``On the use of
  watermark-based schemes to detect cyber-physical attacks,'' \emph{EURASIP
  Journal on Information Security}, vol. 2017, no.~1, 2017.

\bibitem{Weerakkody2014}
S.~Weerakkody, Y.~Mo, and B.~Sinopoli, ``Detecting integrity attacks on control
  systems using robust physical watermarking,'' in \emph{53rd IEEE Conference
  on Decision and Control (CDC)}, Los Angeles, California, 2014, pp.
  3757--3764.

\bibitem{omur_smartgridcomm}
O.~Ozel, S.~Weerakkody, and B.~Sinopoli, ``Physical watermarking for securing
  cyber-physical systems via packet drop injections,'' in \emph{To appear, 8th
  IEEE International Conference on Smart Grid Communications}, 2017.

\bibitem{Chabukswar2011}
R.~Chabukswar, Y.~Mo, and B.~Sinopoli, ``Detecting integrity attacks on {SCADA}
  systems,'' in \emph{18th {IFAC} World Congress}, Milan, Italy, Aug 2011, pp.
  11\,239--11\,244.

\bibitem{schenato2007foundations}
L.~Schenato, B.~Sinopoli, M.~Franceschetti, K.~Poolla, and S.~S. Sastry,
  ``Foundations of control and estimation over lossy networks,''
  \emph{Proceedings of the IEEE}, vol.~95, no.~1, pp. 163--187, 2007.

\bibitem{mo2013lqg}
Y.~Mo, E.~Garone, and B.~Sinopoli, ``{LQG} control with {M}arkovian packet
  loss,'' in \emph{Control Conference (ECC), 2013 European}.\hskip 1em plus
  0.5em minus 0.4em\relax IEEE, 2013, pp. 2380--2385.

\bibitem{moscs10security}
Y.~Mo and B.~Sinopoli, ``False data injection attacks in cyber physical
  systems,'' in \emph{First Workshop on Secure Control Systems}, Stockholm,
  Sweden, April 2010.

\bibitem{PasqualettiJournal}
F.~Pasqualetti, F.~Dorfler, and F.~Bullo, ``Attack detection and identification
  in cyber-physical systems,'' \emph{IEEE Transactions on Automatic Control},
  vol.~58, no.~11, pp. 2715--2729, 2013.

\bibitem{teixeira2012}
A.~Teixeira, D.~Perez, H.~Sandberg, and K.~H. Johannson, ``Attack models and
  scenarios for networked control systems,'' in \emph{Proceedings of the 1st
  international conference on High Confidence Networked Systems}, Beijing,
  China, 2012, pp. 55--64.

\bibitem{chonavel2002statistical}
T.~Chonavel and J.~Ormrod, \emph{Statistical Signal Processing: Modelling and
  Estimation}, ser. Advanced Textbooks in Control and Signal Processing.\hskip
  1em plus 0.5em minus 0.4em\relax SPRINGER VERLAG GMBH, 2002.

\bibitem{delsarte1978orthogonal}
P.~Delsarte, Y.~Genin, and Y.~Kamp, ``Orthogonal polynomial matrices on the
  unit circle,'' \emph{IEEE Transactions on Circuits and Systems}, vol.~25,
  no.~3, pp. 149--160, 1978.

\end{thebibliography}

\section{Appendix: Proof of Lemma \ref{lem:L0}}
\begin{proof}
We begin with the following Lemma. 
\begin{lemma} \label{lem:zeronorm}
Assume $\{\eta_k\}$ is a stationary Markovian drop process.  Suppose $\alpha>0$ and $\beta>0$ are chosen so that the system has finite cost $J_{(m)}$ \cite{mo2013lqg}[Theorem 3] in the absence of a Gaussian watermark.  Consider $\bar{x}_{k+1} = (A+\eta_k B L_{(m)}) \bar{x}_{k}$ and $\bar{x}_{k+1}' = (A+\eta_k B L_{(m)}) \bar{x}_{k}'$, 
\begin{equation}
\bar{x}_0 = \begin{cases} x_0^0 & \eta_{-1} = 0 \\  x_0^1 & \eta_{-1} = 1 \end{cases}, ~~~~~\bar{x}_0' = \begin{cases} x_0^{0'} & \eta_{-1} = 0 \\  x_0^{1'} & \eta_{-1} = 1 \end{cases}.
\end{equation}
Then, we have
\begin{align*}
\lim_{k \rightarrow \infty} \mathbb{E}[\bar{x}_k^i \bar{x}_k^{j \prime}|\eta_{k-1} = 0] &= 0, ~\lim_{k \rightarrow \infty} \mathbb{E}[\bar{x}_k^i \bar{x}_k^{j \prime}|\eta_{k-1} = 1] = 0, \\
\lim_{k \rightarrow \infty} \mathbb{E}[\bar{x}_k^{i} \bar{x}_k^{j\prime}] &= 0, ~~ \lim_{k \rightarrow \infty} \mathbb{E}[(\bar{x}_k^{i})^2] = 0,
\end{align*}
where $i,j \in \{1,\cdots,n\}$ and $\bar{x}_k^i$ is the $i$th element of $\bar{x}_k$ and $\bar{x}_k^{j \prime}$ is the $j$th element of $\bar{x}_k'$.
\end{lemma}
\begin{proof}
Define $\mathcal{J}_N \triangleq \sum_{k = 0}^N \epsilon \mathbb{E}[ \bar{x}_k^T \bar{x}_k]$ where $\epsilon > 0$ is chosen so $\epsilon I \le W$. Moreover the cost to go function is defined as $\bar{V}_j(\bar{x}_j) \triangleq \sum_{k = j}^N \epsilon \mathbb{E}[ \bar{x}_k^T \bar{x}_k|\eta_{-1:j-1}]$. We show that
\begin{equation}
\bar{V}_k(\bar{x}_k) = \begin{cases} \mathbb{E}[\bar{x}_k^T \bar{S}_k \bar{x}_k|\eta_{-1:k-1}] & (\eta_{k-1} = 0) \\ \mathbb{E}[\bar{x}_k^T \bar{R}_k \bar{x}_k|\eta_{-1:k-1}] & (\eta_{k-1} = 1) \end{cases}, \label{eq:detval}
\end{equation}
where
\begin{align*}
&\bar{S}_k =  g_1(\bar{S}_{k+1},\bar{R}_{k+1}), \bar{R}_k =  g_2(\bar{S}_{k+1},\bar{R}_{k+1}),  \\
 &g_1(X,Y) \triangleq \epsilon I + \bar{\alpha} A^TXA + \alpha F^TYF, \\
 &g_2(X,Y) \triangleq \epsilon I +  \beta A^TXA + \bar{\beta} F^TYF,
\end{align*}
and $F = A+BL_{(m)}$, $\bar{S}_N = \epsilon I$, $\bar{R}_N = \epsilon I$. The proof is by induction. \eqref{eq:detval} holds for $k = N$. Assume it holds for $k = t+1$. We next show \eqref{eq:detval} holds for $k = t$. Conditioned on $\eta_{t-1} = 0$, we have
\begin{align*}
&\bar{V}_t(\bar{x}_t) = \mathbb{E}[\epsilon \bar{x}_t^T\bar{x}_t +  \bar{V}_{t+1}(\bar{x}_{t+1})|\eta_{-1:t-1}], \\
&= \mathbb{E}[\epsilon \bar{x}_t^T{x}_t +  \bar{x}_{t}^T (\bar{\alpha} A^T \bar{S}_{t+1} A + \alpha  F^T \bar{R}_{t+1} F ) \bar{x}_t|\eta_{-1:t-1}], \\
&= \mathbb{E}[  \bar{x}_t^T g_1(\bar{S}_{t+1},\bar{R}_{t+1}) \bar{x}_t|\eta_{-1:t-1}] = \mathbb{E}[\bar{x}_t^T \bar{S}_t \bar{x}_t|\eta_{-1:t-1}].
\end{align*}
The case when $\eta_{t-1} = 1$ is similar. Thus,
\begin{equation}
\mathcal{J}_N  = \mathbb{E}[\bar{V}_0(\bar{x}_0)] = \frac{\beta {x_0^0}^T \bar{S}_0 x_0^{0} + \alpha {x_0^1}^T \bar{R}_0 x_0^{1}}{\alpha + \beta}.
\end{equation}
We claim that $\lim_{N \rightarrow \infty} \mathcal{J}_N$ exists. Consider the sequence
\begin{equation}
\mathcal{S}_{k+1} = g_1(\mathcal{S}_{k},\mathcal{R}_{k}), \mathcal{R}_{k+1} = g_2(\mathcal{S}_{k},\mathcal{R}_{k}), \mathcal{R}_0 = \mathcal{S}_0 = \epsilon I.
\end{equation}
We observe that $g_1(X,Y)$ and $g_2(X,Y)$ are monotonically increasing functions in $(X,Y)$. Because $\mathcal{S}_{1} \ge \mathcal{S}_{0}$ and $\mathcal{R}_{1} \ge \mathcal{R}_{0}$, we see that $\{\mathcal{S}_k\}$ and $\{\mathcal{R}_k\}$ are monotonically increasing in the semidefinite sense. Now consider the sequence
\begin{equation}
\bar{\mathcal{S}}_{k+1} = h_1(\bar{\mathcal{S}}_{k},\bar{\mathcal{R}}_{k}), \bar{\mathcal{R}}_{k+1} = h_2(\bar{\mathcal{S}}_{k},\bar{\mathcal{R}}_{k}), \bar{\mathcal{R}}_0 = \bar{\mathcal{S}}_0 = \epsilon I,
\end{equation}
where we define
\begin{align*} 
 h_1(X,Y) \triangleq W + \alpha L_{(m)}^T U L_{(m)} + \bar{\alpha} A^TXA + \alpha F^TYF, \\
 h_2(X,Y) \triangleq W + \bar{\beta} L_{(m)}^T U L_{(m)} + \beta A^TXA + \bar{\beta} F^TYF.
\end{align*}
Again, we observe that $h_1(X,Y)$ and $h_2(X,Y)$ are monotonically increasing in $(X,Y)$. Because $\bar{\mathcal{S}}_{1} \ge \bar{\mathcal{S}}_{0}$ and $\bar{\mathcal{R}}_{1} \ge \bar{\mathcal{R}}_{0}$, it can be seen that $\{\bar{\mathcal{S}}_k\}$ and $\{\bar{\mathcal{R}}_k\}$ are monotonically increasing in the semidefinite sense. Moreover, due to Lemma 4 in \cite{mo2013lqg}, $\{\bar{\mathcal{S}}_k\}$ and $\{\bar{\mathcal{R}}_k\}$ converge. We observe that if $X \le \bar{X}$ and $Y \le \bar{Y}$, then $g_1(X,Y) \le h_1(\bar{X},\bar{Y})$ and $g_2(X,Y) \le h_2(\bar{X},\bar{Y})$. Since $\mathcal{R}_0 = \bar{\mathcal{R}}_0$ and $\mathcal{S}_0 = \bar{\mathcal{S}}_0$, it can be seen that $\mathcal{S}_k \le \bar{\mathcal{S}}_k$ and $\mathcal{R}_k \le \bar{\mathcal{R}}_k$ for all $k$.

As a result, $\{\mathcal{S}_k\}$ and $\{\mathcal{R}_k\}$ are bounded above by a monotonically increasing, convergent sequence, which in turn means that $\{\mathcal{S}_k\}$ and $\{\mathcal{R}_k\}$ are bounded. From the monotone convergence theorem $\{\mathcal{S}_k\}$ and $\{\mathcal{R}_k\}$ converge to some $\mathcal{S}^*$ and $\mathcal{R}^*$. It is immediately seen that
\begin{equation}
\lim_{N \rightarrow \infty} \mathcal{J}_N = \frac{\beta {x_0^0}^T \mathcal{S}^* x_0^{0} + \alpha {x_0^1}^T \mathcal{R}^*  x_0^{1}}{\alpha + \beta}.
\end{equation}
Note that $\mathcal{J}_N = \sum_{k = 0}^N \epsilon \mathbb{E}[\bar{x}_k^T \bar{x}_k]$. Since $\mathcal{J}_N$ converges to a finite constant, $\lim_{k \rightarrow \infty} \mathbb{E}[\bar{x}_k^T \bar{x}_k] = 0.$  Since $0 \le \mathbb{E}[(\bar{x}_k^i)^2] \le \mathbb{E}[\bar{x}_k^T \bar{x}_k]$, we immediately obtain
\begin{equation}
\lim_{k \rightarrow \infty} \mathbb{E}[(\bar{x}_k^i)^2] = 0.
\end{equation}
By symmetry this also implies that $\lim_{k \rightarrow \infty} \mathbb{E}[(\bar{x}_k^{j\prime})^2] = 0.$ By Cauchy Schwartz, we see
\begin{equation}
0 \le (\mathbb{E}[\bar{x}_k^i\bar{x}_k^{j\prime}])^2 \le  \mathbb{E}[(\bar{x}_k^i)^2]\mathbb{E}[(\bar{x}_k^{j\prime})^2].
\end{equation}
Since $\mathbb{E}[(\bar{x}_k^i)^2]\mathbb{E}[(\bar{x}_k^{j\prime})^2]$ converges to 0, we know $(\mathbb{E}[\bar{x}_k^i\bar{x}_k^{j\prime}])^2$ converges to 0 and thus
\begin{equation}
\lim_{k \rightarrow \infty} \mathbb{E}[\bar{x}_k^i\bar{x}_k^{j\prime}] = 0. \label{eq:uncondlim0}
\end{equation}
Next, we observe that
\begin{equation}
0 \le \frac{\min(\alpha,\beta)}{\alpha+\beta} \mathbb{E}[(\bar{x}_k^i)^2 |\eta_{k-1} = l] \le \mathbb{E}[(\bar{x}_k^i)^2],
\end{equation}
where $l \in \{0,1\}$. As $\alpha,\beta > 0$ by assumption, we know $\lim_{k \rightarrow \infty}  \mathbb{E}[(\bar{x}_k^i)^2 |\eta_{k-1} = l] = 0$. Using the Cauchy Schwartz inequality in a similar manner as before, we see
\begin{equation}
\lim_{k \rightarrow \infty} \mathbb{E}[\bar{x}_k^i\bar{x}_k^{j\prime}|\eta_{k-1} = 1] = 0, \lim_{k \rightarrow \infty} \mathbb{E}[\bar{x}_k^i\bar{x}_k^{j\prime}|\eta_{k-1} = 0] = 0. \label{eq:limcor0}
\end{equation}
\end{proof}
We are now ready to prove the desired result. To this end, we first observe that
\begin{equation}
\begin{pmatrix} \mathbb{E}[\bar{x}_{k+1}' \bar{x}_{k+1}^T|\eta_k = 0]  \\  \mathbb{E}[\bar{x}_{k+1}' \bar{x}_{k+1}^T|\eta_k = 1] \end{pmatrix} = \mathcal{L}_0 \begin{pmatrix} \mathbb{E}[\bar{x}_{k}' \bar{x}_{k}^T|\eta_{k-1} = 0]  \\  \mathbb{E}[\bar{x}_{k}' \bar{x}_{k}^T|\eta_{k-1} = 1] \end{pmatrix}.
\end{equation}
As a result,
\begin{equation}
 \begin{pmatrix} \mathbb{E}[\bar{x}_{k}' \bar{x}_{k}^T|\eta_{k-1} = 0]  \\  \mathbb{E}[\bar{x}_{k}' \bar{x}_{k}^T|\eta_{k-1} = 1] \end{pmatrix} = \mathcal{L}_0^k \begin{pmatrix} x_0^{0'} x_0^{0~T}  \\  x_0^{1'} x_0^{1~T}  \end{pmatrix}.
\end{equation}
Leveraging \eqref{eq:limcor0}, we see that $\lim_{k \rightarrow \infty} \mathbb{E}[\bar{x}_{k}' \bar{x}_{k}^T| \eta_{k-1} = l] = 0$ for $l \in \{0,1\}$. Consequently, we have
\begin{equation}
\lim_{k \rightarrow \infty} \mathcal{L}_0^k \begin{pmatrix} x_0^{0'} x_0^{0~T}  \\  x_0^{1'} x_0^{1~T}  \end{pmatrix}  = 0. \label{eq:lim0}
\end{equation}
Note that $x_0^{0'}, x_0^0, x_0^{1'},$ and $x_0^1$ can be chosen so that $\mathcal{L}_ 0^k$ is applied to an arbitrary canonical basis vector in $\mathbb{R}^{2n \times n}$. Thus, for all $M \in \mathbb{R}^{2n \times n}$, $\lim_{k \rightarrow \infty} \mathcal{L}_0^k(M) = 0$. Thus, $\mathcal{L}_0$ is stable.
\end{proof}

\section{Appendix: Proof of Theorem \ref{thm:2ndwatermark}}
\begin{proof}
We begin with the following Lemma. 
\begin{lemma}
Suppose $p_d$ is chosen so the system with IID drops has finite cost $J_{(b)}$ \cite{mo2013lqg}[Theorem 3]. Then the matrix $(A+\bar{p}_dBL_{(b)})$ is Schur stable. Moreover, the operator $\mathcal{L}_1(X) \triangleq \bar{p}_d(A+BL_{(b)})X(A+BL_{(b)})^T + p_dAXA^T$ is stable. Specifically, $\forall M \in \mathbb{R}^{n \times n}$, we have $\lim_{k \rightarrow \infty} \mathcal{L}_1^k(M) = 0$
\end{lemma}
\begin{proof}
Consider the systems $\bar{x}_{k+1} = (A+ \eta_k BL_{(b)}) \bar{x}_k$, and $\bar{x}_{k+1}' = (A+ \eta_k BL_{(b)}) \bar{x}_k'$. where $\eta_k$ is an IID drop process with drop probability $p_d$ and $\bar{x}_0 = x_{0,*}, \bar{x}_0' = x_{0,*}'$. Observe that
\begin{equation}
\mathbb{E}[\bar{x}_k] = (A+ \bar{p}_d BL_{(b)})^k x_{0,*}.
\end{equation}
Noting that the IID drop case is a special instance of Markovian drops, we know from Lemma \ref{lem:zeronorm} that $\lim_{k \rightarrow \infty} \mathbb{E}[(\bar{x}_k^{i})^2] = 0$. Using the fact that $\mathbb{E}[(\bar{x}_k^{i})^2] \ge (\mathbb{E}[\bar{x}_k^{i}])^2 \ge 0$, we have $\lim_{k \rightarrow \infty} \mathbb{E}[\bar{x}_k] = 0$. As a result, for all $x_{0,*} \in \mathbb{R}^n$
\begin{equation}
\lim_{k \rightarrow \infty} (A+ \bar{p}_d BL_{(b)})^k x_{0,*} = 0.
\end{equation}
Thus, $(A+\bar{p}_dBL_{(b)})$ is Schur stable. Next, we note that
\begin{equation}
\mathbb{E}[\bar{x}_{k+1}'\bar{x}_{k+1}^T] = \mathcal{L}_1(\mathbb{E}[\bar{x}_{k}'\bar{x}_{k}^T]).
\end{equation}
As a result,
\begin{equation}
\mathbb{E}[\bar{x}_{k}'\bar{x}_{k}^T] = \mathcal{L}_1^k(x_{0,*}'x_{0,*}^T).
\end{equation}
Leveraging \eqref{eq:uncondlim0}, we note $\lim_{k \rightarrow \infty}  \mathbb{E}[\bar{x}_{k}' \bar{x}_{k}^T] = 0$ and this implies
\begin{equation}
\lim_{k \rightarrow \infty}  \mathcal{L}_1^k(x_{0,*}'x_{0,*}^T)  = 0. \label{eq:lim1}
\end{equation}
Note that $x_{0,*}'$ and $x_{0,*}$ can be chosen so that $\mathcal{L}_ 1^k$ is applied to an arbitrary canonical basis vector in $\mathbb{R}^{n \times n}$. Thus, for all $M \in \mathbb{R}^{n \times n}$, $\lim_{k \rightarrow \infty} \mathcal{L}_1^k(M) = 0$. Thus, $\mathcal{L}_1$ is stable.
\end{proof}
We now proceed to the main proof. We obtain an equivalent realization to \eqref{eq:opt2} by using autocovariance functions $\Gamma(d) \triangleq \mathbb{E}[\Delta u_k \Delta u_{k+d}]$. \\
\textit{Step 1: Calculate $\bar{J}$ in terms of $\Gamma(d)$}: \\
Let us first compute
\begin{equation*}
\mathbb{E}[ x_t^T W x_t + u_{t,c}^T U u_{t,c} ] = \mbox{tr}(W \mbox{Cov}(x_t)) + \mbox{tr}(U \mbox{Cov}(u_{t,c})),
\end{equation*} 
for fixed $t \ge 0$. It can be seen that 
\begin{align}
&x_t = l_{1,\{\eta_{k}\}}(w_{-\infty:t-1},v_{-\infty:t-1}) + \gamma_t(\Delta u_{-\infty:t-1}), \label{eq:x1} \\
&u_{t,c} = l_{2,\{\eta_{k}\}}(w_{-\infty:t-1},v_{-\infty:t}) + \eta_t L_{(b)} \gamma_t(\Delta u_{-\infty:t-1}) \nonumber \\ &+ \eta_t \Delta u_t,\nonumber \\
& \gamma_t(\Delta u_{-\infty:t-1}) =  \sum_{i = 1-t}^{\infty} \bigg[ \prod_{j = 1 - i}^{t-1} (A + \eta_j BL_{(b)}) \bigg] \eta_{-i} B \Delta u_{-i}, \nonumber
\end{align}
where $l_1$ and $l_2$ are linear functions of the process and sensor noise for fixed realizations of the drop process $\eta_k$. Since $\{w_k\}$ and $\{v_k\}$ are independent of $\Delta u_k$, we observe that
\begin{align}
\bar{J} &= J_{(b)}(p_d) + \frac{1}{N} \lim_{N \rightarrow \infty} \sum_{t = 0}^{N-1} \Bigg( \mbox{tr}(W \mbox{Cov}(\gamma_t)) + \nonumber  \\
& \mbox{tr}(U \mbox{Cov}(\eta_t [L_{(b)} \gamma_t + \Delta u_t])) \Bigg) . \label{eq:Jappend}
\end{align}
\textit{Step 1a: Calculate} $\mbox{Cov}(\gamma_t)$: \\
Define $Z_t \triangleq \sum_{i = 1-t}^\infty  \gamma_{i,t} \gamma_{i,t}^T$ where
\begin{equation*}
 \gamma_{i,t} \triangleq \bigg[\prod_{j = 1-i}^{t-1} (A + \eta_j BL_{(b)})\bigg]\eta_{-i} B \Delta u_{-i}.
\end{equation*}
We see that $\mathbb{E}[Z_{t+1}]$ is equal to 
\begin{displaymath}
\mathbb{E}[(A+\eta_t BL_{(b)})Z_t(A+\eta_t BL_{(b)})^T + \eta_t^2 B \Delta u_t \Delta u_t^T B^T].
\end{displaymath}
Since $\eta_t$ is independent of $Z_t$, we have
\begin{equation*}
\mathbb{E}[Z_{t+1}] = \mathcal{L}_1(\mathbb{E}[Z_{t}]) + \bar{p}_d B\Gamma(0)B^T.
\end{equation*}
Since $\mathcal{L}_1$ is stable and the system has been running since $k = -\infty$, $\mathbb{E}[Z_t]$ is the unique solution of the following fixed point equation.
\begin{align*}
E[Z_t] = \mathcal{L}_1(\mathbb{E}[Z_{t}]) + \bar{p}_d B\Gamma(0)B^T = L_1(B\Gamma(0)B^T).     
\end{align*}
In addition, let 
\begin{align*}
Y_t^d  &\triangleq \sum_{i = 1-t}^\infty  \xi_{i,t}^d \gamma_{i,t}^T,  \\
\xi_{i,t}^d &\triangleq  \bigg[\prod_{j = 1-i-d}^{t-1} (A + \eta_j BL_{(b)})\bigg]\eta_{-i-d} B \Delta u_{-i-d}.
\end{align*}
By similar reasoning, we find that $E[Y_t^d]$ equals
\begin{equation*}
 L_1\left(\bar{p}_d(A+BL_{(b)})(A+\bar{p}_d BL_{(b)})^{d-1} B \Gamma(d) B^T \right).
\end{equation*}
We argue
\begin{align}
\mbox{Cov}(\gamma_t) &= \mathbb{E}\left[Z_t + \sum_{d = 1}^{\infty} Y_t^d + (Y_t^d)^T\right] \nonumber  \nonumber\\   
&= 2 \sum_{d = 1}^{\infty} \mbox{sym}[Y_*^d]  +  L_1(B\Gamma(0)B^T), \label{eq:covJ1}
\end{align}
where 
\begin{equation*}
 Y_*^d = L_1\left(\bar{p}_d(A+BL_{(b)})(A+\bar{p}_d BL_{(b)})^{d-1} B \Gamma(d) B^T \right).
\end{equation*}
\textit{Step 1b: Calculate} $\mbox{Cov}(\eta_t [L_{(b)} \gamma_t + \Delta u_t])$: \\
We argue that 
\begin{equation*}
\mathbb{E}[\eta_t^2 L_{(b)} \gamma_t \Delta u_t^T] = \bar{p}_d^2 L_{(b)} \sum_{d = 0}^{\infty} (A+\bar{p}_dBL_{(b)})^d B \Gamma(d+1).
\end{equation*}
Therefore, we obtain
\begin{align}
& \mbox{Cov}(\eta_t [L_{(b)} \gamma_t + \Delta u_1]) = \bar{p}_d\left(\Gamma(0) + L_{(b)} \mbox{Cov}(\gamma) L_{(b)}^T\right)  \nonumber\\
& + 2 \mbox{sym} \left(\bar{p}_d^2 L_{(b)} \sum_{d = 0}^{\infty} (A+\bar{p}_dBL_{(b)})^d  B \Gamma(d+1) \right) \label{eq:covJ2} , 
\end{align}
where $\mbox{Cov}(\gamma) \triangleq \mbox{Cov}(\gamma_t)$ is given in \eqref{eq:covJ1}. Note, $\mbox{Cov}(\gamma_t)$  is constant in $t$. Substituting \eqref{eq:covJ2} into \eqref{eq:Jappend}, we have $\bar{J}=$
\begin{align*}
&J_{(b)}(p_d) + \mbox{tr}(\bar{p}_d U \Gamma(0)) + \mbox{tr}((W +\bar{p}_d L_{(b)}^TUL_{(b)})\mbox{Cov}(\gamma)) \nonumber  \\
&+ \mbox{tr}\left(2 U \mbox{sym} \bigg(\bar{p}_d^2 L_{(b)} \sum_{d = 0}^{\infty} (A+\bar{p}_dBL_{(b)})^d B \Gamma(d+1) \bigg)\right).
\end{align*}

\textit{Step 2: Calculate $\mathbb{E}[y_k^T y_k' |\mathcal{H}_0]$ in terms of $\Gamma(d)$}: \\
Recall from the proof of Theorem \ref{thm:correlation} and \eqref{eq:correlation}
\begin{equation}
\mathbb{E}[y_k^T y_k'|\mathcal{H}_0] = \mbox{tr} \left(C\mathbb{E}[x_k'x_k^T]C^T\right)  \nonumber.
\end{equation}
We observe that $x_k' = \gamma_k$. Thus, from \eqref{eq:x1}, we assert
\begin{equation}
\lim_{k \rightarrow \infty} \mathbb{E}[y_k^T y_k'|\mathcal{H}_0] =  \mbox{tr} \left(C\mbox{Cov}(\gamma)C^T\right).
\end{equation}

\textit{Step 3: Convert to Frequency Domain}: \\
Optimizing over the autocovariance functions is intractable as there are infinitely many optimization variables. In this case, as in the work \cite{Mo2014}, we will leverage Bochner's theorem in \cite[p.64]{chonavel2002statistical} (see also \cite{delsarte1978orthogonal}). This theorem provides a frequency domain representation of an autocovariance function of a stationary process:
\begin{theorem}[Bochner's theorem] 
$\Gamma(d)$ is an autocovariance function of a stationary Gaussian process $\{\Delta u_k\}$ if and only if there exists a unique positive Hermitian measure $\nu$ of size $p \times p$ satisfying
\begin{equation*}
\Gamma(d) = \int_{-0.5}^{0.5} \exp(2\pi j d \omega) \mbox{d}\nu(\omega).
\end{equation*}
\end{theorem}
Note that a positive Hermitian measure $\nu$ takes a Borel set in $[-0.5~0.5]$ and outputs a positive semidefinite Hermitian matrix in $\mathbb{C}^{p \times p}$. We choose to optimize over $\tilde{\Gamma}(d)$, which has bijective relationship with $\Gamma(d)$. By assumption $\tilde{\Gamma}(d)$ is an autocovariance function of a stationary Gaussian process. As a result, we can use Bochner's theorem to rewrite $\tilde{\Gamma}(d)$ in terms of a Riemann sum. Specifically,
\begin{equation}
\tilde{\Gamma}(d) = \lim_{\sigma \rightarrow 0} 2 \mbox{Re} \left[ \sum_{i = 1}^q \exp(2 \pi j d \omega_i) \tilde{\nu}(I_i) \right],
\end{equation}
where $I_i \cap I_j = \emptyset,~~ \cup_{i = 1}^q I_i = [0,~0.5], ~~\omega_i \in I_i$
and $\sigma$ is the maximum length of $I_i$. Here, we  also leverage the fact that $\tilde{\Gamma}(d)$ is real. Moreover, from \eqref{eq:covJ1}, we see that
\begin{align*}
&\mbox{Cov}(\gamma) \\
&= \lim_{\sigma \rightarrow 0} \sum_{i=1}^q \Big( 2\mbox{Re} \Big[ 2 \mbox{sym} \Big( L_1 \Big[ \bar{p}_d \exp(2 \pi j \omega_i) \bar{\rho} (A+BL_{(b)}) \nonumber \\
& \sum_{d=1}^{\infty} (\bar{\rho} \exp(2\pi j\omega_i) (A+\bar{p}_dBL_{(b)}))^{d-1} B \tilde{\nu}(I_i) B^T \Big] \Big) \\
& + L_1 \Big[B \tilde{\nu}(I_i) B^T \Big] \Big] \Big),
\end{align*}
\begin{align}
&= \lim_{\sigma \rightarrow 0} \sum_{i=1}^q \Big( 2\mbox{Re} \Big[ 2 \mbox{sym} \Big( L_1 \Big[ \bar{p}_d \exp(2 \pi j \omega_i) \bar{\rho} (A+BL_{(b)}) \nonumber \\
& (I - \bar{\rho} \exp(2\pi j\omega_i) (A+\bar{p}_dBL_{(b)}))^{-1} B \tilde{\nu}(I_i) B^T \Big] \Big) \nonumber \\
& + L_1 \Big[B \tilde{\nu}(I_i) B^T \Big] \Big] \Big), \nonumber \\
&= \lim_{\sigma \rightarrow 0} \sum_{i=1}^q F_2(\omega_i, \tilde{\nu}(I_i), p_d).
\end{align}
The inverse is well defined since we showed $(A+\bar{p}_dBL_{(b)})$ is Schur stable. By similar reasoning it can be shown that
\begin{equation}
\bar{J} = J_{(b)}(p_d) + \lim_{\sigma \rightarrow 0} \sum_{i=1}^q F_1(\omega_i, \tilde{\nu}(I_i), p_d).
\end{equation}
Replacing $\rho(A_\omega) \le \bar{\rho}$ with Assumption 1 in problem \eqref{eq:opt2}, we arrive at the following equivalent formulation:
\begin{equation}
\begin{aligned}
& \underset{\tilde{\nu}({I}_i),p_d}{\text{maximize}}
& & \lim_{\sigma \rightarrow 0} \sum_{i=1}^q \mbox{tr}(C F_2(\omega_i, \tilde{\nu}(I_i), p_d) C^T) \\
& \text{subject to}
& & J_{(b)}(p_d) + \lim_{\sigma \rightarrow 0} \sum_{i=1}^q F_1(\omega_i, \tilde{\nu}(I_i), p_d) \leq \delta, \\
& & & 0 \leq p_d \leq 1. \label{eq:opt4}
\end{aligned}
\end{equation}
\textit{Step 4: Demonstrate Equivalence}: \\
The rest of the result follows from Steps 3 and 4 in the proof of Theorem 6 in \cite{Mo2014} when $p_d < 1$. In particular, we can leverage the linearity of $F_2$ and $F_1$ in $H$ for fixed $p_d < 1$ and $\omega$ to show that the optimal value of \eqref{eq:opt3} is an upper bound on the optimal value for problem \eqref{eq:opt4}. Then, we show that for Borel set $S_b \subset [-0.5,~0.5]$, the measure
\begin{equation}
\tilde{\nu}(S_b) = \mathbb{I}_{\omega_{*} \in S_b} H_* + \mathbb{I}_{-\omega_{*} \in S_b} \mbox{conj}(H_*),
\end{equation}
where $\mathbb{I}$ is the indicator function and $\mbox{conj}$ refers to the complex conjugate, achieves this upper bound. The resulting autocovariance function is
\begin{equation}
\Gamma(d) = 2 \bar{\rho}^{|d|} \mbox{Re}(\exp(2 \pi j d \omega_{*}) H_*),
\end{equation}
and can be generated by the HMM \eqref{eq:optHMM} if there exists an  optimal $H_*$, which has rank 1. Theorem 7 of \cite{Chabukswar2013} demonstrates the existence of such a solution, while the associated proof shows how such a solution can be constructed from an  optimal $H_*$ with rank greater than 1. When $p_d = 1$, $F_2$ and $F_1$ are identically 0, establishing the equivalence of \eqref{eq:opt3} and \eqref{eq:opt4} in this scenario. Note also in this case, (if $J_{(b)}(p_d) \le \delta$) any stationary Gaussian process in the feasible region is optimal since the resulting additive input is immediately dropped. \end{proof}


%

\end{document}